\documentclass[journal]{IEEEtran}

\usepackage{ifpdf}
\usepackage{cite}
\ifCLASSINFOpdf
  \usepackage[pdftex]{graphicx}
\else
\fi
\usepackage{graphicx}
\usepackage{subfigure}
\usepackage{amsmath}
\usepackage{amsthm}
\usepackage{amssymb}
\usepackage{amsfonts}
\usepackage{amsmath,bm}
\usepackage{algorithmic}
\usepackage{algorithm}
\usepackage{array}
\usepackage{mathrsfs}
\usepackage{float}
\usepackage{placeins}
\usepackage{afterpage}
\usepackage{array} 
\usepackage{booktabs}
\usepackage{multirow, listings}
  \usepackage{verbatim}
\graphicspath{IEEEtran}
\newtheorem{theo}{Theorem}

\usepackage{subfigure,cleveref}
\begin{document}

\title{Waveform Design for 6G ISAC Systems Under Full-Duplex Residual Self-Interference}

\author{Ning Wei, Aimin Tang,~\IEEEmembership{Senior~Member,~IEEE,} Yin Xu,~\IEEEmembership{Senior~Member,~IEEE,} and Wenze Qu
\thanks{Ning Wei, Aimin Tang, and Yin Xu are with Shanghai Jiao Tong University, Shanghai, China. Wenze Qu is with the MediaTek Inc. Beijing, China. Corresponding author: Aimin Tang (email: tangaiming@sjtu.edu.cn).}
}

\maketitle

\begin{abstract}
In this paper, the waveform design for 6G integrated sensing and communication (ISAC) systems is investigated, with a particular focus on the practical limitations imposed by imperfect full-duplex radios. Under such imperfections, continuous communication waveforms, such as orthogonal frequency-division multiplexing (OFDM), suffer from severe full-duplex residual self-interference (RSI) for radar sensing, which significantly restricts the long-range sensing capabilities required by emerging low-altitude wireless networks (LAWN). To address this challenge, we propose a novel time-division ISAC waveform that integrates a specially developed dual-power phase-coded pulse for sensing into the communication frame under full-duplex RSI. Specifically, the dual-power sensing pulse consists of a high-power sequence followed by a low-power sequence, effectively exploiting imperfect full-duplex operations to achieve reliable long-range sensing while eliminating the detection blind range inherent to conventional half-duplex pulse radars. Furthermore, a complementary and inverse-phase sequence group is designed to ensure perfect autocorrelation and robust cross-correlation sidelobe suppression, so as to enhance multi-target detection capability. As for sensing signal processing, a parameterized mismatched filter is developed and optimized to maximize the detection performance, tailored to the proposed pulse structure. In addition, we design a hierarchical one-dimensional constant false-alarm-rate cell-average (CFAR-CA) detector that can exploit the perfect range-domain autocorrelation characteristics of the proposed waveform to further improve the detection performance. 
Extensive simulations demonstrate that the proposed design significantly improves the maximum detection range and multi-target detection capability compared to existing OFDM and linear frequency-modulated (LFM) pulse baselines, while effectively covering the blind range for targets with small radar cross section (RCS).
\end{abstract}


%
\IEEEpeerreviewmaketitle

\section{Introduction}
Low-altitude wireless networks (LAWN) are the fundamental infrastructure to support the low-altitude economy \cite{yuan2025ground}. Connected aerial robots, including manned and unmanned aerial vehicles (UAVs), are key players in low-altitude intelligent transportation in the LAWN \cite{UAV,DroneRef,10403776}. To achieve robust and efficient management of aerial robots, both communication and radar sensing functions are indispensable. In recent years, integrated sensing and communication (ISAC), which enables radar sensing and wireless communication to share the same hardware architecture and signal processing blocks, has attracted significant research interest for 6G networks \cite{survey1,survey2,3GPP,survey3,Weisurvey,9509294}. Therefore, the ISAC design can naturally align communication and sensing requirements in LAWN on a unified platform. 

The realization of integrated sensing functionality in an ISAC base station (BS) necessitates simultaneous signal transmission and target echo reception, requiring full-duplex operation. However, insufficient isolation between transmitter and receiver introduces self-interference that influences the received sensing signals. Although multiple techniques, including antenna isolation, analog cancellation, and digital cancellation, can suppress self-interference, it is hard to achieve a perfect full-duplex radio in practical implementations \cite{SIC2,SIC3,tang2014balanced,SIlimited,SIlimited2}. Residual self-interference (RSI) elevates the system noise floor, potentially obscuring weak echoes from long-range targets and consequently degrading the long-range detection performance. However, long-range sensing capability represents a critical requirement for LAWN. Thus, under imperfect full-duplex radios, the ISAC waveform design for efficient sensing poses a key technical challenge.

The communication-centric waveform design is recognized as a promising approach for 6G ISAC networks. 
This strategy effectively incorporates sensing capabilities while maintaining compatibility with current communication systems, requiring only minimal modifications.
Extensive studies have focused on communication-centric waveform design for ISAC systems \cite{5GPRS,SSB,zhao2023reference,tang2024isi,zhou2024improving,liu2025cp, Constellation}. Among various studies, the orthogonal frequency-division multiplexing (OFDM) waveform has been established as the most widely adopted solution due to its compatibility with existing communication systems and adaptability to sensing requirements. 
Based on the OFDM waveform, there are two typical approaches to achieve wireless sensing: pilot-based design, such as \cite{5GPRS,SSB,zhao2023reference,tang2024isi,zhou2024improving}, and data payload-based design, such as \cite{liu2025cp,Constellation}.
Pilot symbols enable efficient sensing by providing deterministic reference signals (RSs), thereby simplifying implementation. Some studies exploit the existing RS, such as the positioning reference signal (PRS) \cite{5GPRS} and synchronization signal block (SSB) \cite{SSB}, for sensing. However, the sensing performance may not be guaranteed in this way. By considering both the sensing  and communication requirements, the RS is designed and optimized in \cite{zhao2023reference}. For long-range sensing, the OFDM waveform also suffers the critical problem of inter-symbol interference (ISI), due to the limited duration of the cyclic prefix (CP). To address this problem, an RS design with zero-power RS deployment is developed in \cite{tang2024isi}, which can highly improve the ISI-free sensing range. In \cite{zhou2024improving}, the RS design enabling alternating CP and cyclic postfix is designed to combat the ISI for long-range sensing. Besides pilot-based sensing design, data payload-based sensing provides a higher processing gain. However, the data payload-based design requires the elimination of communication data randomness during sensing processing, which necessitates real-time storage of stochastic communication data and consequently introduces additional system overhead. According to the analysis in \cite{liu2025cp}, the cyclic prefix-OFDM (CP-OFDM) waveform achieves superior ranging sidelobe suppression compared with other commonly used communication waveforms. 
In \cite{Constellation}, an optimal probabilistic constellation shaping approach is proposed to achieve an adjustable tradeoff between sensing and communication performance.  
Although various designs have been proposed, the impact of RSI is not addressed in the aforementioned studies. When the imperfect full-duplex radio is considered, the RSI can highly degrade the performance of the above designs, especially for long-range sensing. To address the impact of RSI, a pulse-like radar is incorporated into the OFDM waveform by a tailored pilot design and corresponding signal processing method in \cite{tang2021self} with the IEEE 802.11ad system. The proposed design is resistant to the RSI, so that effective long-range sensing can be achieved. However, such a design relies on the extremely large OFDM subcarrier spacing provided in the IEEE 802.11ad system, which cannot be directly applied to the 6G system. 

Unlike continuous waveforms that suffer from the influence of RSI, pulsed waveforms provide an effective solution to the self-interference problem due to half-duplex operation. However, half-duplex operation inherently leads to the partial eclipsing effect, which results in a detection blind range [\citen{radarbook}, Chap.~5].
Moreover, the pulse-based ISAC solution can lead to a low communication rate. To improve the communication data rate, a masked modulation scheme is developed in \cite{masked} based on half-duplex radios, which can highly improve the communication rate (up to 50\% duty cycle for communications) with pulsed radar systems. However, this design still suffers from the blind-range problem. In \cite{11373582}, the pulse waveform is inserted in the communication frame in a time-division manner, and the proposed design relies on a coordinated multi-BS ISAC framework to eliminate near-field blind zones, which, however, highly improves the system complexity. 
In \cite{waveform2022}, a time-division waveform is developed for ISAC based on full-duplex radios, where pulses are utilized for radar sensing, and communication transmission occupies the silent time between pulses. It can provide a higher communication rate, while the sensing performance is limited by the self-interference cancellation (SIC) capability.


In the aforementioned context, we propose a novel time-division waveform design for 6G ISAC systems with imperfect full-duplex radios. Specifically, we design a phase-coded pulse for sensing, while retaining conventional OFDM for communications. By leveraging the opportunities brought by the imperfect full-duplex radio, a special dual-power pulse structure with a high-power sequence component followed by a low-power sequence component is designed. The proposed waveform eliminates the blind range and achieves robust long-range sensing. To further enhance multi-target detection performance, we introduce a sequence group design that effectively suppresses correlation sidelobes. To fully exploit the potential of our designed waveform, a corresponding mismatched filter and a hierarchical one-dimensional (1D) detector are further designed for sensing signal processing. 
Comprehensive simulations validate that the proposed design significantly outperforms conventional linear frequency-modulated (LFM) and OFDM schemes. It can achieve a minimum detectable radar-cross-section (RCS) below $-40$ dBsm for short-range targets and maintains reliable long-range sensing independent of the system's SIC capability. With our proposed sequence design, it can successfully resolve closely spaced targets with a $20$ dB RCS disparity, mitigating the weak-target masking issue. Furthermore, the proposed hierarchical 1D detector achieves a substantially longer maximum detectable range than traditional two-dimensional (2D) methods.

The main contributions of this paper are summarized as follows:
\begin{itemize}
    \item Considering a practical imperfect full-duplex radio with limited self-interference cancellation, we propose a novel pulse sensing waveform composed of a high-power component followed by a low-power component, together with a mode-switching strategy between half-duplex and full-duplex reception. This design enables reliable long-range sensing while simultaneously eliminating the blind range inherent to traditional pulse radars.
    \item Based on the developed pulse structure, we construct the sequence group using complementary and inverse-phase sequence pairs, which can achieve perfect autocorrelation in the range domain, while effectively suppressing cross-correlation sidelobes between high- and low-power sequences. 
    \item A mismatched filter is developed to fully exploit the potential of the proposed waveform. It employs parallel processing with a delay-dependent weighting, allowing the low-power sequence to enhance sensing performance while maintaining low computational complexity. Within the constant false-alarm-rate cell-averaging (CFAR-CA) detection framework, we introduce comprehensive performance metrics, including the signal-to-sidelobe-plus-interference-plus-noise ratio (SSINR) for partial eclipsing scenarios. The weighting parameter is subsequently optimized based on the proposed metrics.
    \item A tailored hierarchical 1D CFAR-CA detection scheme is developed. Leveraging the near-elimination of range-domain sidelobes by our designed pulse sequences, the proposed detector achieves superior performance compared to the conventional 2D range-Doppler domain detection.
\end{itemize}

The rest of this paper is organized as follows. The system model is presented in Section \ref{section 2}. The sensing waveform and corresponding sequence design are elaborated in Section \ref{section 3}. The mismatched filtering design and target detection algorithm are presented in Section \ref{section 4}, with subsequent parameter optimization in Section \ref{section 5}. Simulation results are provided in Section \ref{section 6}. This paper is concluded in Section \ref{section 7}.


\begin{figure*}[t]
    \centering
    \includegraphics[width=0.98\textwidth]{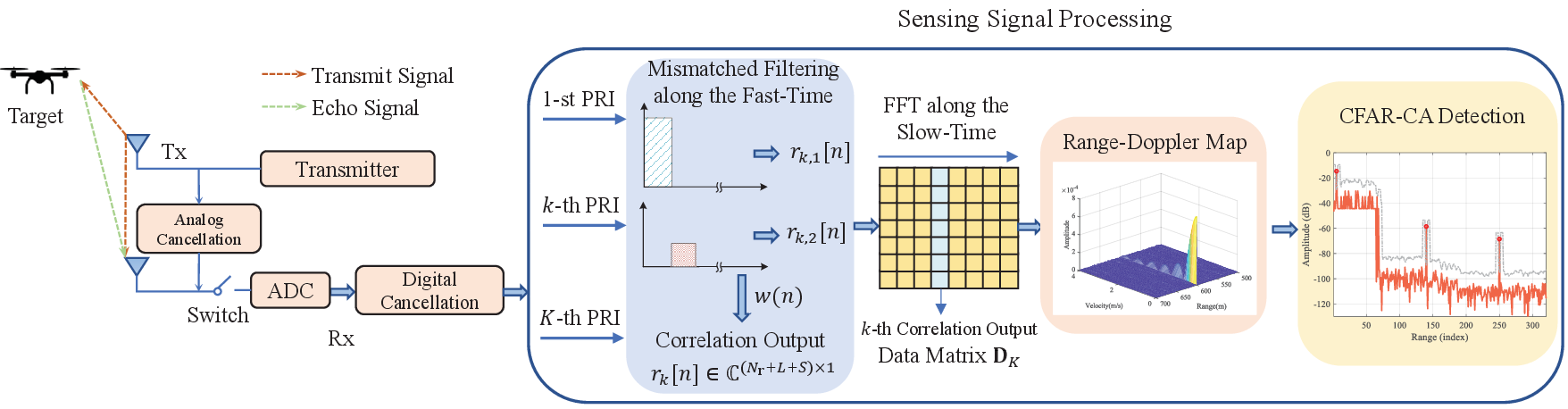 }
    \caption{System architecture and the sensing signal processing framework.}
    \label{system model}
    \vspace{-0em}
\end{figure*}

\textbf{Notations:}
$x$, $\mathbf{x}$, $\mathbf{X}$ represent a scalar, a vector, and a matrix, respectively. $\mathbf{x}^\mathsf{T}$ is the transpose of $\mathbf{x}$, and $\mathbf{x}^\mathsf{H}$ is the Hermitian transpose of $\mathbf{x}$. The notations $\mathbf{0}_{M\times N}$, $\mathbf{1}_{M\times N}$, and $\mathbf{I}_m$ denote the $M\times N$ all-zero matrix, $M\times N$ all-one matrix, and identity matrix of order $m$ respectively. The operator $\|\cdot\|_p$ represents the $\ell_p$-norm of its argument.
Operator $\mathrm{diag}(\mathbf{x})$ denotes the diagonal matrix formed by the elements of vector $\mathbf{x}$. 

\section{System Model}\label{section 2}



A monostatic ISAC system is considered, where the transmitter and receiver are co-located and integrated into a single transceiver, as shown in Fig. \ref{system model}. The ISAC BS is equipped with separate transmit and receive antennas, enabling simultaneous signal transmission and radar echo reception. However, an imperfect full-duplex radio is considered for the ISAC BS. More specifically, although antenna isolation, analog cancellation, and digital cancellation are all applied, 
the total SIC capability has an upper bound. Thus, when the transmission power is very high, strong RSI can still exist after cancellation. The RSI remains due to the finite SIC capability, which represents an inherent limitation in practical full-duplex systems. Moreover, the receiver's radio-frequency (RF) chain and analog-to-digital converter (ADC) will function correctly only if the transmission power is properly controlled at a low level; excessive power will saturate them. To support a high-power transmission for long-range sensing, a switch is added in the receiver RF chain, which can enable the imperfect full-duplex receiver to also support half-duplex mode. Therefore, the conventional high-power pulse transmission under half-duplex mode can be supported in our system model.  

Since this work concentrates on the waveform design and the associated signal processing algorithm, a single-antenna transceiver structure is employed to maintain clarity of presentation without loss of generality. The proposed design and conclusion can be readily extended to multi-antenna system implementations.

\section{Waveform Design}\label{section 3}


In this paper, a time-division ISAC waveform is designed, where our developed pulse waveform is periodically inserted into the traditional OFDM communication waveform. Thus, there is no modification to the existing communication waveform. In the following part of this paper, we will elaborate on the pulse waveform design for sensing and its corresponding signal processing design.

\subsection{Pulse Waveform Design for Sensing}
In conventional radar, the pulse waveform is a typical choice for long-range sensing via half-duplex radios. Due to its half-duplex operation, there exists a minimum range for the conventional pulse radar [\citen{radarbook2}, Chap.~3]
\begin{equation}
    R_{\min}=\frac{c_0(\tau+T_{\mathrm{r}})}{2},
    \label{blind}
\end{equation}
where $c_0$ represents the speed of light, $\tau$ denotes the pulse duration, and $T_{\mathrm{r}}$ is the recovery time for switching from transmission to reception. The minimum range in Eq. (\ref{blind}) is also referred to as the blind range in some textbooks such as [\citen{radarbook}, Chap.~5]. When a target lies within the radar’s minimum range, the receiver either cannot obtain the echo or receives only a portion of it, a condition known as partial eclipsing. Consequently, the accumulated echo energy for a target in the blind range is not sufficient for reliable target detection.

In this paper, we aim to design a novel pulse waveform that leverages the advantages of both pulse radar and full-duplex radio, while accounting for the practical limitations of imperfect full-duplex transceivers. The proposed pulse waveform is illustrated in Fig. \ref{signal model}, which utilizes a dual-power transmission scheme consisting of high-power and low-power components. The high-power component provides sufficient energy for long-range detection, while the low-power component eliminates the blind range to maintain continuous detection coverage.
The proposed sensing waveform consists of four consecutive phases, namely the high-power transmission period $T_{\mathrm{h}}$ with a transmission power of $P_{\mathrm{h}}$, the recovery period $T_{\mathrm{r}}$, the low-power transmission period $T_{\mathrm{l}}$ with a transmission power of $P_{\mathrm{l}}$, and the silent period $T_{\mathrm{s}}$. To preserve the interference-free advantage of the silent period in conventional radar, we have $T_{\mathrm{s}}\gg T_{\mathrm{h}}+T_{\mathrm{r}}+T_{\mathrm{l}}$. The total waveform duration is expressed as $T_{\mathrm{t}}=T_{\mathrm{h}}+T_{\mathrm{r}}+T_{\mathrm{l}}+T_{\mathrm{s}}$. 

At the receiver side, a switch is used to control the radio operating in half-duplex mode or full-duplex mode. When the high-power component is transmitted, the system works in half-duplex mode to avoid the saturation of the receiver RF chain. After $T_{\text{off}} = T_{\mathrm{h}} + T_{\mathrm{r}}$, the system is switched to full-duplex radio to receive the target echo with a duration of $T_{\text{on}} = T_{\mathrm{t}} - T_{\text{off}}$. 
For a target in the minimum range, due to the half-duplex operation for the high-power component, no samplings or partial samplings can be received; however, the low-power component of the target can be fully received due to full-duplex operation, although it suffers from RSI. Fortunately, the target in the minimum range is usually with a high echo power, which can be used to combat self-interference. For a target in the long range with a weak echo power, it is not affected by the RSI, and a full echo of the high-power component can be received for target detection. That is the basic principle of our proposed design for reliable long-range sensing, while eliminating the blind range in traditional pulse radar.     

\begin{figure}[t]
    \centering
    \includegraphics[width=0.75\linewidth]{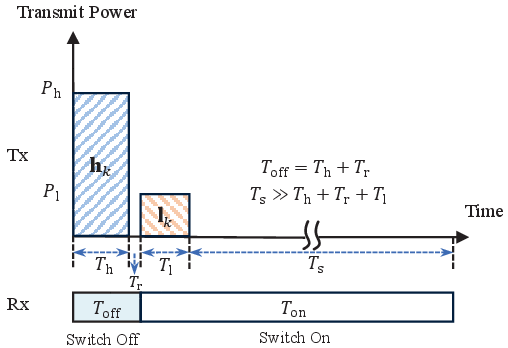 }
    \caption{The structure of our proposed sensing waveform.}
    \label{signal model}
    \vspace{-0em}
\end{figure}


The proposed pulse structure can be periodically inserted into the OFDM communication waveform in a time-division way, as illustrated in Fig. \ref{frame}. Assuming that the pulse repetition interval is $T$ and a coherent processing interval (CPI) contains $K$ pulses, the mathematical model for our proposed sensing pulse can be expressed as 
\begin{equation} \label{transmit signal}
    x(t)=\sum_{k=0}^{K-1}x_k(t-kT), \quad 0 \leq t \leq KT,
\end{equation}
where $x_k(t)$ represents the sensing waveform transmitted during the $k$-th pulse repetition interval (PRI), which can be further expressed as
\begin{equation} 
    x_k(t)=
    \begin{cases}
    \sqrt{P_\mathrm{h}}h_k(t),\hfill  0<t\leq T_\mathrm{h} \\
    0,\hfill T_\mathrm{h}<t\leq T_\mathrm{h}+T_\mathrm{r}\\
    \sqrt{P_\mathrm{l}}l_k(t-T_\mathrm{h}-T_\mathrm{r}),\hfill T_\mathrm{h}+T_\mathrm{r}<t\leq T_\mathrm{h}+T_\mathrm{r}+T_\mathrm{l} \\
    0,\hfill T_\mathrm{h}+T_\mathrm{r}+T_\mathrm{l}<t\leq T_\mathrm{t}
    \end{cases}
    \label{k-th transmit}
\end{equation}
where $h_k(t)$ and $l_k(t)$ denote the normalized waveforms for the high-power and low-power components, respectively, and they satisfy
\begin{equation} \label{energy1}
    \frac{1}{T_\mathrm{h}}\int_0^{T_\mathrm{h}}\lvert h_k(t)\lvert^2\mathrm{d}t=1,
\end{equation}
\begin{equation} \label{energy2}
    \frac{1}{T_\mathrm{l}}\int_{0}^{T_\mathrm{l}}\lvert l_k(t)\lvert^2\mathrm{d}t=1.
\end{equation}
The phased coded pulse waveform is considered in our design. Thus, $h_k(t)$ and $l_k(t)$ are generated by pulse shaping their respective discrete sequences. Specifically, $h_k(t)$ is formed from the length-$H$ sequence $\mathbf{h}_k=\big[h_{0,k},\ldots,h_{H-1,k}\big]^\mathsf{T}\in\mathbb{C}^{H\times1}$, whereas $l_k(t)$ is constructed from the length-$L$ sequence $\mathbf{l}_k=\big[l_{0,k},\ldots,l_{L-1,k}\big]^\mathsf{T}\in\mathbb{C}^{L\times1}$. The resulting continuous-time waveform constructed through pulse shaping can be expressed as
\begin{equation}
    h_{k}\left(t\right)=\sum_{i=0}^{H-1}h_{i,k}\varphi\left(t-iT_{\mathrm{p}}\right),\quad 0<t\leq T_{\mathrm{h}},
    \label{high signal}
\end{equation}
\begin{equation}
    l_{k}\left(t\right)=\sum_{i=0}^{L-1}l_{i,k}\varphi\left(t-iT_{\mathrm{p}}\right),\quad 0<t\leq T_{\mathrm{l}},
    \label{low signal}
\end{equation}
where $T_{\mathrm{p}}=\frac{1}{B}$ denotes the chip duration with a bandwidth of $B$, $T_{\mathrm{h}}=HT_{\mathrm{p}}$ and $T_{\mathrm{l}}=LT_{\mathrm{p}}$ represent the durations of the high-power and low-power components, respectively, and $\varphi(t)$ denotes the pulse-shaping filter.
According to Eq. (\ref{energy1}) and Eq. (\ref{energy2}), the sequences $\mathbf{h}_k$ and $\mathbf{l}_k$ are required to satisfy the following constraints:
\begin{equation}
    \left\|\mathbf{h}_k\right\|^2=H,\quad \left\|\mathbf{l}_k\right\|^2=L.
\end{equation}
To limit the peak-to-average power ratio (PAPR) performance, the proposed design employs constant-modulus sequences that satisfy $|h_{i,k}|=1$ and $|l_{i,k}|=1$. 

\begin{figure}[t]
    \centering
    \includegraphics[width=0.45\textwidth]{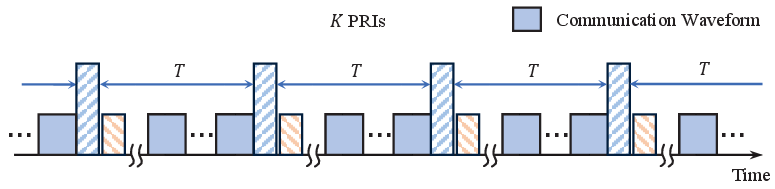 }
    \caption{Configuration of communication and sensing waveforms.}
    \vspace{-0em}
    \label{frame}
\end{figure}

\subsection{Sequence Design}
With our proposed phased-coded pulse structure, the sequences $\mathbf{h}_k$ and $\mathbf{l}_k$ can be designed to have excellent autocorrelation characteristics with low sidelobes, thereby enabling superior sensing performance for multi-target detection. Within the framework of the phase-coded pulse waveform, several established coding schemes exist, including biphase Barker codes for fixed low sidelobe levels, minimum peak sidelobe (MPS) codes, maximal length sequences with pseudorandom characteristics, Frank codes with good periodic correlation, etc. However, a single unimodular sequence is inherently unable to achieve perfect autocorrelation properties \cite{7420715}. Fortunately, the radar receiver usually jointly processes multiple pulses for target detection. In this case, complementary sequences \cite{1972complementary} have been employed in the design of phase-coded waveforms \cite{9133306,10681468}, which can perfectly eliminate the range sidelobes.
Therefore, we adopt complementary sequence pairs for the high-power component and the low-power component separately. Assuming that $\mathbf{a}_{H}$ and $\mathbf{b}_{H}$ are a complementary sequence pair for the high-power component and $\mathbf{a}_{L}$ and $\mathbf{b}_{L}$ are a complementary sequence pair for the low-power component, we can get two pulses with distinct sequences $(\mathbf{a}_{H},\mathbf{a}_{L})$ and $(\mathbf{b}_{H},\mathbf{b}_{L})$.

Although complementary sequences can achieve perfect autocorrelation, the proposed pulse structure introduces a critical challenge. Specifically, the use of a high-power component followed by a low-power component leads to non-negligible cross-correlation between these two parts. As shown in Fig.  \ref{Fig:cross-correlation}, the existence of correlation can highly degrade the detection performance for a target that falls into the cross-correlation region. 

\begin{figure}[t]
    \centering
    \subfigure[]
    {\includegraphics[width = 0.48\linewidth]{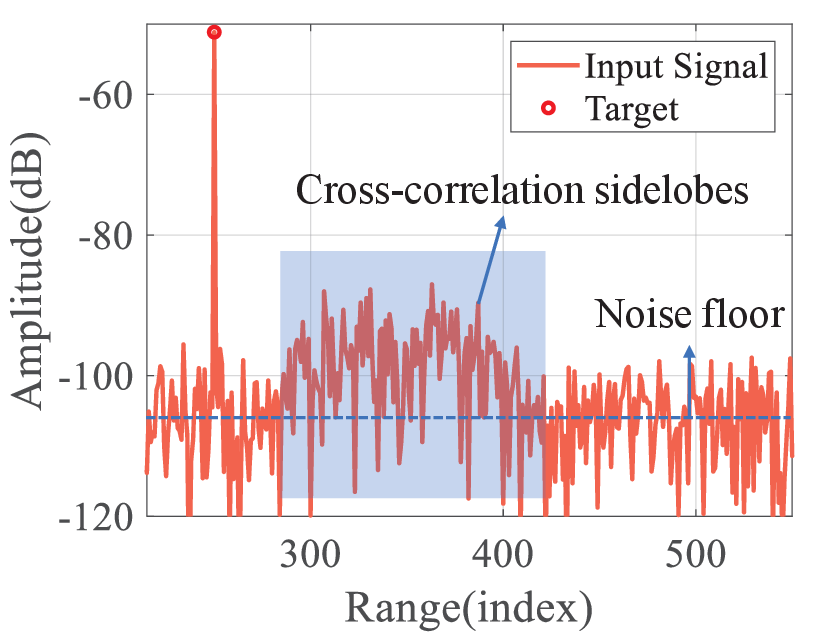}\label{Fig:cross-correlation}}
    \subfigure[]
    {\includegraphics[width=0.48\linewidth]{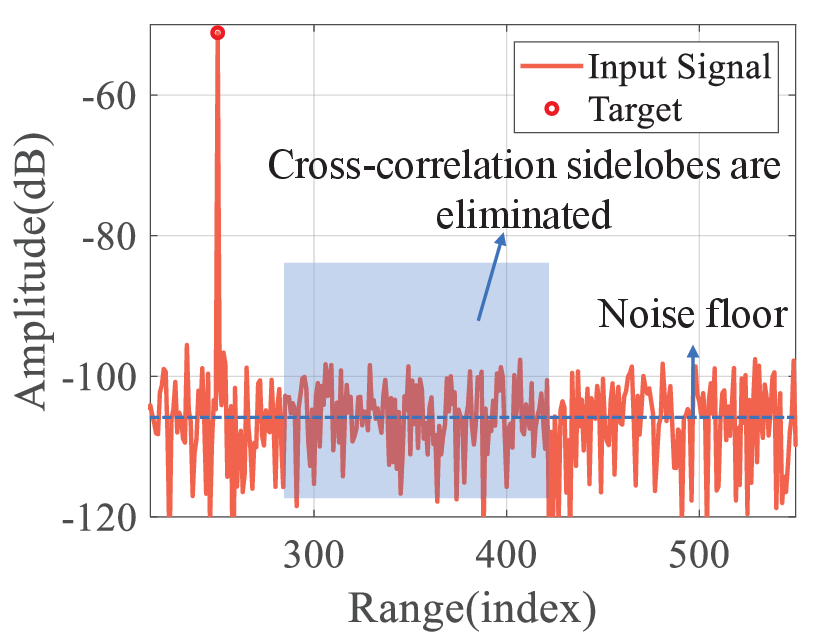}\label{Fig:cross-correlation2}}
    \caption{Cross-correlation sidelobes using (a) only complementary sequences and (b) the proposed set $\mathcal{C}$.
}
    \vspace{-0em}
\end{figure}

To address the cross-correlation problem, we introduce inverse-phase sequences and construct a sequence set $\mathcal{C}$ based on complementary sequences, defined as
\begin{equation}
    \mathcal{C} = \{(\mathbf{a}_{H},\mathbf{a}_{L}), (\mathbf{b}_{H},\mathbf{b}_{L}),(\mathbf{a}_{H},-\mathbf{a}_{L}),(\mathbf{b}_{H},-\mathbf{b}_{L})\},
\end{equation}
where the sequences $-\mathbf{a}_{L}$ and $-\mathbf{b}_{L}$ represent the negations of the sequences $\mathbf{a}_{L}$ and $\mathbf{b}_{L}$, respectively. 
For real-valued sequences composed of $\pm1$ elements, this negation operation is equivalent to scalar multiplication by $-1$. In the case of complex-valued sequences, the negation corresponds to a phase shift of $\pi$ radians applied to each complex sample in the sequence, mathematically expressed as $e^{j\pi}\mathbf{a}_L = -\mathbf{a}_L$ and $e^{j\pi}\mathbf{b}_L = -\mathbf{b}_L$. It should be noted that inverse-phase sequences can be applied either to the high-power component or the lower-power component, and the sequence set $\mathcal{C}$ adopts the latter case as an example. In addition, since the proposed sequence set $\mathcal{C}$ comprises four pulses, to preserve the designed correlation properties during joint processing, the number of PRIs $K$ within each CPI should satisfy $K \equiv 0 \pmod{4}$.

With the above design, we can have the following properties. Since $\mathbf{a}_{H}$ ($\mathbf{a}_{L}$) and $\mathbf{b}_{H}$ ($\mathbf{b}_{L}$) are a complementary sequence pair, we have the following autocorrelation result:  
\begin{equation}
     R_{\mathbf{a}_{H}}(\tau) + R_{\mathbf{b}_{H}}(\tau) = 2H\delta(\tau),
     \label{highsq}
\end{equation}
\begin{equation}
     R_{\mathbf{a}_{L}}(\tau) + R_{\mathbf{b}_{L}}(\tau) = 2L\delta(\tau),
\end{equation}
where $R_{\mathbf{x}}$ represent the aperiodic autocorrelation function of sequence $\mathbf{x}$ and $\delta(\tau)$ denotes the Dirac delta function.
With the introduction of the $(\mathbf{a}_{H},-\mathbf{a}_{L})$ and $(\mathbf{b}_{H},-\mathbf{b}_{L})$, the cross-correlation sidelobes between the high-power and low-power sequences can be effectively eliminated: 
\begin{equation}
    \begin{aligned}
        R_{\mathbf{a}_H,\mathbf{a}_L}(\tau)+R_{\mathbf{a}_{H},(-\mathbf{a}_{L})}(\tau)=0, \quad  \forall \tau,\\
        R_{\mathbf{b}_H,\mathbf{b}_L}(\tau)+R_{\mathbf{b}_{H},(-\mathbf{b}_{L})}(\tau)=0, \quad \forall \tau,
    \end{aligned}
    \label{negation}
\end{equation}
where $R_{\mathbf{x},\mathbf{y}}(\tau)$ denotes the aperiodic cross-correlation function (ACCF) between two sequences $\mathbf{x}$ and $\mathbf{y}$.
As shown in Eq. (\ref{negation}), the joint processing of a sequence group enables the mutual cancellation of cross-correlation sidelobes generated by high- and low-power sequences, which is also demonstrated in Fig. \ref{Fig:cross-correlation2}. The sidelobe cancellation mechanism relies on the phase opposition between cross-correlation sidelobes. However, the negation operation maintains the complementary characteristics, i.e.,
\begin{equation}
    R_{(-\mathbf{a}_{L})}(\tau) + R_{(-\mathbf{b}_{L})}(\tau) = 2L\delta(\tau).
\end{equation}
Accordingly, the proposed sequence set $\mathcal{C}$ guarantees perfect autocorrelation performance under joint processing while simultaneously achieving effective suppression of cross-correlation sidelobes.

\subsection{Received Waveform in Sensing Duration}
For each pulse, the receiver begins to sample after $T_{\text{off}}$. Assuming there are $I$ sensing targets, the received waveform can be expressed as
\begin{equation}
    \begin{aligned}
    y(t)=\sum_{i=1}^I{\alpha_i} x(t-\tau_i)e^{j2\pi f_{d,i}t}+\beta g(t)+z(t),\\
    kT+T_{\text{off}}\leq t\leq kT+T_{\text{t}}, k=0,\cdots,K-1
    \end{aligned}
    \label{echo}
\end{equation}
where $\beta g(t)$ characterizes the RSI after SIC, thereby reflecting the limited SIC capability of the full-duplex radio, and $\beta$ denotes the associated SIC capability; $z(t)$ is the noise; $\tau_i=\frac{2R_i}{c_0}$ represents the round-trip propagation delay associated with the $i$-th target located at range $R_i$; $f_{d,i}=\frac{2f_c v_i}{c_0}$ is the Doppler shift induced by the $i$-th target's radial velocity $v_i$ with $f_c$ denoting the carrier frequency; and ${\alpha_i}$ denotes the complex channel coefficient for target $i$. The magnitude of ${\alpha_i}$ usually adopts the classical two-way radar propagation model, expressed as $ |\alpha_i|^2=\frac{G_\mathrm{t}G_\mathrm{r}\lambda^2\sigma_i}{(4\pi)^3R_i^4}$, where $G_{\mathrm{t}}$ and $G_{\mathrm{r}}$ denote the transmit and receive antenna gains, $\lambda$ is the carrier wavelength, and $\sigma_i$ is the RCS of the target. Note that the maximum detectable delay satisfies $0<\tau\leq T_{\mathrm{t}}-T_{\mathrm{h}}$. Moreover, to prevent interference from the low-power component in the subsequent communication symbol interval, we assume that $T_{\mathrm{l}}\leq T_{\mathrm{h}}$, i.e., $L\leq H$.

Considering the received signal during the $k$-th PRI, the signal model can be re-expressed as
\begin{equation}
    y_k(t)=\sum_{i=1}^I\alpha_i x_k(t-\tau_i)e^{j2\pi f_{d,i}t}e^{j2\pi f_{d,i}kT}+\beta g_k(t)+z_k(t),
    \label{PRI}
\end{equation}
where $T_{\mathrm{off}}< t \leq T_{\mathrm{t}}$. Within each PRI of duration $T_{\mathrm{t}}$, the complex exponential term $e^{j2\pi f_{d,i} t}$ in Eq. (\ref{PRI}) can be treated as approximately constant because the Doppler-induced phase variation over a single PRI is negligible, i.e., 
$f_{d,i}t\ll 1$ and thus $e^{j2\pi f_{d,i} t}\approx 1$. 
This property implies that reliable Doppler estimation requires processing across the slow-time dimension over multiple PRIs.
Therefore, the sensing received signal during the $k$-th PRI in Eq. (\ref{PRI}) can be simplified as
\begin{equation}
    y_k(t)=\sum_{i=1}^I\alpha_i x_k(t-\tau_i)e^{j2\pi f_{d,i}kT}+\beta g_k(t)+z_k(t).
\end{equation}
The system noise $z_k(t)$ is usually modeled as zero-mean additive white Gaussian noise (AWGN) with circularly symmetric complex statistics, given by $z_k(t)\sim\mathcal{CN}(0, N_{0}B)$, where $N_{0}$ denotes the noise power spectral density (PSD) and $B$ denotes the system bandwidth. In addition, the RSI is also modeled as zero-mean AWGN, as used in \cite{RSImodel1, RSImodel2, RSImodel3}. Therefore, given a limited SIC cancellation capability $\beta$, the RSI $\beta g_k(t)$ can be further expressed as a linear function of the transmission power of the low-power component as $\beta \sqrt{P_{\mathrm{l}}}p(t)s_k(t)$, where  $s_k(t)\sim \mathcal{CN}(0,1)$ and $p(t)$ is the rectangular window function, which is defined as
\begin{equation}
    \begin{aligned}        p(t)&=\mathbb{I}_{[T_{\text{off}},T_{\text{off}}+T_{\mathrm{l}})}(t)\\
        &=\begin{cases}
            1, & \mathrm{for~}T_{\text{off}}\leq t<T_{\text{off}}+T_{\mathrm{l}} \\
            0, & \mathrm{otherwise}  
        \end{cases}
    \end{aligned}
\end{equation}
Thus, the received signal after SIC can be expressed as
\begin{equation}
    y_k(t)=\sum_{i=1}^I\alpha_i x_k(t-\tau_i)e^{j2\pi f_{d,i}kT}+\beta \sqrt{P_{\mathrm{l}}}p(t)s_k(t)+z_k(t).
    \label{receive signal}
\end{equation}


Since the multi-target scenario constitutes a straightforward extension of the single-target case, the following analysis is carried out for brevity, i.e., $I=1$. Moreover, we assume the target is located in the sampling bin $n_\tau$ to simplify the presentation of the discrete expression. If the target is not in the sampling bin, our proposed signal processing design can also be readily applied. 
The sampled signal can be expressed in vector form as
\begin{equation}
    \mathbf{y}_{k}=\alpha e^{j2\pi f_{d}kT}\mathbf{x}_{k,n_{\tau}}+\beta\mathbf{P}_{\text{SI}}\mathbf{s}_{k}+\mathbf{P}_{\text{z}}\mathbf{z}_{k},
    \label{Rxsignal}
\end{equation}
where $\mathbf{y}_k \in \mathbb{C}^{(H+N_{\mathrm{r}}+L+S)\times 1}$ denotes the received sampling vector, 
$S=T_{\mathrm{s}}/T_\mathrm{p}$ and $N_{\mathrm{r}}=T_{\mathrm{r}}/T_\mathrm{p}$ represent the numbers of samples collected during 
the silent and recovery intervals, respectively, $n_\tau=1,2,\ldots,N_{\mathrm{r}}+L+S$ follows from the delay constraint 
$0<\tau \leq T-T_{\mathrm{h}}$, and $\mathbf{x}_{k,n_\tau}$ represents the target echo within the 
$k$-th PRI. The matrix $\mathbf{P}_{\text{SI}}=\operatorname{diag}(\mathbf{p}_{\text{SI}})$ specifies the power distribution of self-interference, with $\mathbf{p}_{\text{SI}}=[\mathbf{0}_{1\times (H+N_{\mathrm{r}})}, \sqrt{P_\mathrm{l}}\mathbf{1}_{1\times L}, \mathbf{0}_{1\times S}]^\mathsf{T}$.
The matrix $\mathbf{P}_{\text{z}}=\operatorname{diag}(\mathbf{p}_{\text{z}})$ specifies the power distribution of noise, with $\mathbf{p}_{\text{z}}=[\mathbf{0}_{1\times (H+N_{\mathrm{r}})},  \mathbf{1}_{1\times (L+S)}]^\mathsf{T}$.
The terms $\mathbf{s}_{k}$ and $\mathbf{z}_{k}$ represent complex Gaussian random vectors following $\mathbf{s}_{k} \sim \mathcal{CN}(\mathbf{0},\mathbf{I}_{H+N_{\mathrm{r}}+L+S})$ and $\mathbf{z}_{k} \sim \mathcal{CN}(\mathbf{0},N_0B\mathbf{I}_{H+N_{\mathrm{r}}+L+S})$, respectively. 
Furthermore, the condition $T_{\mathrm{s}}\gg T_{\mathrm{h}}+T_{\mathrm{r}}+T_{\mathrm{l}}$ ensures that $S>H+N_{\mathrm{r}}+L$. Based on Eqs.~\eqref{k-th transmit}, \eqref{high signal}, and \eqref{low signal}, the received signal corresponding to a given delay index $n_\tau$ can therefore be written as
\begin{equation}\label{eq:receivedX}
    \begin{aligned}
        &\mathbf{x}_{k,n_{\tau}}=\\
        &\begin{cases}
            \left[\mathbf{0}_{1\times (H+N_{\mathrm{r}}+n_{\tau})},\sqrt{P_{\mathrm{l}}}l_{0,k},\ldots,\sqrt{P_{\mathrm{l}}}l_{L-1,k},\mathbf{0}_{1\times(S-n_{\tau})} \right]^{\mathsf{T}} \\
            \hfill 0< n_{\tau}\leq N_{\mathrm{r}},\\
            
            \left[\mathbf{0}_{1\times (H+N_{\mathrm{r}})},\sqrt{P_{\mathrm{h}}}h_{H-n_{\tau}+N_{\mathrm{r}},k},\ldots,\sqrt{P_{\mathrm{h}}}h_{H-1,k},\mathbf{0}_{1\times N_{\mathrm{r}}} \right.  \\
            \left. \sqrt{P_{\mathrm{l}}}l_{0,k},\ldots,\sqrt{P_{\mathrm{l}}}l_{L-1,k},\mathbf{0}_{1\times(S-n_{\tau})}\right]^{\mathsf{T}}  \hfill N_{\mathrm{r}}< n_{\tau}\leq H+N_{\mathrm{r}}, \\
            
            \left[\mathbf{0}_{1\times n_{\tau}},\sqrt{P_{\mathrm{h}}}h_{0,k},\ldots,\sqrt{P_{\mathrm{h}}}h_{H-1,k}, \mathbf{0}_{1\times N_{\mathrm{r}}}\sqrt{P_{\mathrm{l}}}l_{0,k},\ldots,\right.  \\
            \left. \sqrt{P_{\mathrm{l}}}l_{L-1,k},\mathbf{0}_{1\times(S-n_{\tau})}\right]^{\mathsf{T}}  \hfill H+N_{\mathrm{r}}< n_{\tau}\leq S, \\
            
            \left[\mathbf{0}_{1\times n_{\tau}},\sqrt{P_{\mathrm{h}}}h_{0,k},\ldots,\sqrt{P_{\mathrm{h}}}h_{H-1,k}, 
            \mathbf{0}_{1\times N_{\mathrm{r}}}, \sqrt{P_{\mathrm{l}}}l_{0,k},\ldots,\right.  \\
            \left. \sqrt{P_{\mathrm{l}}}l_{L+S-n_{\tau}-1,k},\right]^{\mathsf{T}}  
            \hfill S< n_{\tau}\leq L+S.\\

            \left[\mathbf{0}_{1\times n_{\tau}},\sqrt{P_{\mathrm{h}}}h_{0,k},\ldots,\sqrt{P_{\mathrm{h}}}h_{H-1,k}, 
            \mathbf{0}_{1\times (L+S+N_{\mathrm{r}}-n_{\tau})}\right]^{\mathsf{T}}  \\
            \hfill L+S< n_{\tau}\leq N_{\mathrm{r}}+L+S.
           
        \end{cases}
    \end{aligned}
\end{equation}

\section{Sensing Signal Processing} \label{section 4}
As shown in Fig. \ref{system model}, the sensing signal processing includes the fast-time processing with a mismatched filter design for delay estimation, the slow-time processing with the conventional fast Fourier transform (FFT) operation for Doppler estimation. After the range-Doppler (RD) map is obtained, the CFAR-CA detector is applied for target detection. The details for sensing signal processing are elaborated in this section.

\subsection{Mismatched Filter Design for Fast-Time Processing}
For pulse radar, the matched filter is usually adopted to maximize the output signal-to-noise ratio (SNR). However, our proposed pulse and the corresponding reception mode switching lead to a special received signal structure: 1) for the target in the minimum/short range (delay less than $T_\text{off}$), the high-power sequence suffers partial/full reception loss, while the low-power sequence is fully received under the impact of RSI; 
2) for the target in the middle range (delay larger than $T_\text{off}$ but less than $T_\text{off}+T_\text{l}$), both sequences can be fully received, but the received high-power sequence suffers from the RSI; 3) for the target in the long range (delay larger than $T_\text{off}+T_\text{l}$), the high-power sequence can be fully received without the impact of RSI. To exploit the feature that the low-power sequence can always be fully received with excellent autocorrelation properties, a mismatched filter is designed in this paper for the fast-time processing. More specifically, a tunable weight $w$ is applied in the receiving processing filtering as
\begin{equation}
    \begin{split}
        \mathbf{f}_k^{\prime}=
        &\left[\sqrt{P_\mathrm{h}}h_{0,k},\dots,\sqrt{P_\mathrm{h}}h_{H-1,k},\mathbf{0}_{1\times N_{\mathrm{r}}}\right. \\
        &\left. w\sqrt{P_\mathrm{l}}l_{0,k},\dots,w\sqrt{P_\mathrm{l}}l_{L-1,k},\mathbf{0}_{1\times S}\right]^\mathsf{T}.
    \end{split}
    \label{allUMF}
\end{equation}
Consequently, the normalized mismatched filter can be expressed as
\begin{equation}
    \mathbf{f}_k=\frac{\mathbf{f}_k^{\prime}}{\|\mathbf{f}_k^{\prime}\|}=\frac{\mathbf{f}_k^{\prime}}
    {\sqrt{P_\mathrm{h}H+w^2P_{\mathrm{l}}L}}.
\end{equation}

The above mismatched filter design can be explained as two matched filters for the high-power sequence and the low-power sequence separately, but a dynamic weight $w$ is designed to balance the contributions of these two parallel matched filters in the final output. This weighting mechanism can effectively suppress the sidelobe elevation caused by partial eclipsing effects during short-range target detection. This tunable weight can be fine-tuned for each potential delay of the target to achieve optimal performance, as presented in the next Section.  
Consequently, the mismatched filter $\mathbf{f}_k$ can be decomposed into
\begin{equation}
    \mathbf{f}_k=\frac{\mathbf{f}_{k,1}^{\prime}+w\mathbf{f}_{k,2}^{\prime}}
    {\sqrt{P_\mathrm{h}H+w^2P_{\mathrm{l}}L}},
    \label{UMF}
\end{equation}
where $\mathbf{f}_{k,1}^{\prime}=[\sqrt{P_\mathrm{h}}h_{0,k},\dots,\sqrt{P_\mathrm{h}}h_{H-1,k},\mathbf{0}_{1\times ( N_{\mathrm{r}}+L+S)}]^\mathsf{T}$ denotes the matched filter for the high-power component, and $\mathbf{f}_{k,2}^{\prime}=[\mathbf{0}_{1\times (H+ N_{\mathrm{r}})},\sqrt{P_\mathrm{l}}l_{0,k},\dots,\sqrt{P_\mathrm{l}}l_{L-1,k},\mathbf{0}_{1 \times S}]^\mathsf{T}$ represents the corresponding matched filter for the low-power component.

The correlation processing output for the fast-time processing in PRI $k$ is obtained by applying the designed mismatched filter $\mathbf{f}_k$ to the discretized received signal $\mathbf{y}_{k}$, which can be mathematically expressed as
\begin{equation}
    r_k[n]=\mathbf{f}_k^{\mathsf{H}}\mathbf{J}_n\mathbf{y}_{k},\quad n=1,\dots,L+S,
    \label{correlation}
\end{equation}
where $\mathbf{J}_n$ represents a shift operator that performs time-shifts on the received signal vector $\mathbf{y}_{k}$. Mathematically, this operator can be expressed as
\begin{equation}
    \mathbf{J}_{n}=\begin{bmatrix}{\mathbf{0}_{(M-n)\times n}}&{\mathbf{I}_{M-n}}\\{\mathbf{0}_{n\times n}}&{\mathbf{0}_{n\times(M-n)}}\end{bmatrix}
\end{equation}
where $M=H+N_{\mathrm{r}}+L+S$ indicates the total number of sampling points.
Based on Eq. (\ref{UMF}), the correlation processing procedure can be simplified to
\begin{equation}\label{eq:corrsum}
    r_k[n]=\frac{1}{\sqrt{P_\mathrm{h}H+w^2P_{\mathrm{l}}L}}\big(r_{k,1}[n]+wr_{k,2}[n]\big),
\end{equation}
where $r_{k,1}[n]=\mathbf{f}_{k,1}^{\prime\mathsf{H}}\mathbf{J}_n\mathbf{y}_{k}$ and $r_{k,2}[n]=\mathbf{f}_{k,2}^{\prime\mathsf{H}}\mathbf{J}_n\mathbf{y}_{k}$ represent the correlation processing outputs of the matched filters $\mathbf{f}_{k,1}^{\prime}$ and $\mathbf{f}_{k,2}^{\prime}$, respectively.

Based on Eq. (\ref{eq:corrsum}), we can see that for different values of $n$, the system can compute the weighted sum of $r_{k,1}[n]$ and $r_{k,2}[n]$ with varying weights $w(n)$, rather than performing multiple correlations using the filter defined in Eq. (\ref{allUMF}). This approach significantly reduces computational complexity while maintaining the desired processing flexibility. The simplified operation requires only a single pair of correlation outputs $\big(r_{k,1}[n],r_{k,2}[n]\big)$ followed by parameterized weighting, eliminating the need for repeated filter implementations across different range indices.

After the fast-time processing with the proposed mismatched filter, each PRI yields a correlation vector of length $N_{\mathrm{r}}+L+S$, corresponding to $N_{\mathrm{r}}+L+S$ discrete delay bins, which can be expressed as
\begin{equation}
    \mathbf{r}_k=\big[r_k[1],r_k[2],\dots,r_k[N_{\mathrm{r}}+L+S]\big]^{\mathsf{T}}\in \mathbb{C}^{(N_{\mathrm{r}}+L+S)\times1}.
\end{equation}
For a CPI comprising $K$ PRIs, this results in a $(N_{\mathrm{r}}+L+S) \times K$ data matrix $\mathbf{D}_K$, i.e., 
\begin{equation}
    \mathbf{D}_K=\left[\mathbf{r}_1,\mathbf{r}_2,\dots,\mathbf{r}_K\right]\in\mathbb{C}^{(N_{\mathrm{r}}+L+S)\times K}.
\end{equation}

\subsection{FFT for Slow-Time Processing}
As established in standard radar signal processing theory, Doppler processing is then applied to each delay bin (i.e., each row of the data matrix $\mathbf{D}_K$) to estimate the Doppler frequency $f_d$ through slow-time processing, and subsequently determine the target radial velocity $v$ through the relationship $v = \lambda f_d/2$, where $\lambda$ denotes the carrier wavelength. 
Doppler processing can be achieved by applying the FFT operation. Specifically, we conduct the $M_\mathrm{FFT}$-point FFT along the row of the data matrix $\mathbf{D}_K$ to generate the RD map as follows:
\begin{equation}
    \mathrm{P}(n,m)=\frac{1}{K}\left|\sum_{k=0}^{M_\mathrm{FFT}-1}(\mathbf{D}_K)_{n,k}e^{-j2\pi\frac{km}{M_\mathrm{FFT}}}\right|^2,
    \label{doppler}
\end{equation}
where $\mathrm{P}(n,m)$ denotes the element on the $n$-th row and $m$-th column of the RD map and $m=\lfloor\frac{-M_{\mathrm{Per}}}{2}\rfloor,\ldots,\lfloor\frac{-M_{\mathrm{Per}}}{2}\rfloor-1$ denotes the Doppler bin.
Doppler detection is periodic in frequency with a principal period ranging from $-\frac{1}{2T}$ to $\frac{1}{2T}$, which implies that Doppler shifts within the range $|f_d| \leq \frac{1}{2T}$ can be unambiguously detected. 

\subsection{Target Detection}
After obtaining RD map, target detection is conducted. Most radar systems employ CFAR detection, which maintains a fixed false-alarm probability by adjusting detection thresholds based on the noise and interference levels. However, the theoretical CFAR method exhibits a few practical limitations, including that the practical environmental noise and dynamic interference level can not be timely estimated, and the existence of sidelobes affects the detection in the multi-target scenarios, e.g, the sidelobes of a strong target may be falsely detected as a target. 
To this end, the CFAR-CA detector is widely used in practical radar systems, which can dynamically adapt the detection threshold based on both the surrounding signal intensity and the prescribed false-alarm probability. This adaptive approach also effectively suppresses sidelobe interference while preserving detection sensitivity for weak targets, particularly under partial eclipsing conditions where elevated sidelobes of the sequence may induce false alarms.

\begin{figure}[t]
    \centering
    \includegraphics[width=0.4\textwidth]{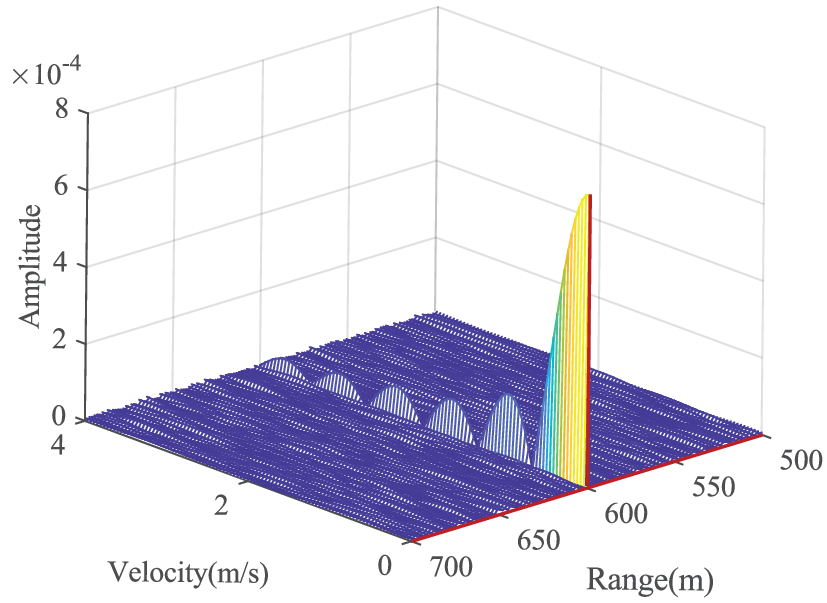 }
    \caption{The RD map example of our proposed waveform for a single target located at 600 m.}
    \label{CFAR-CA}
    \vspace{-0em}
\end{figure}

With the 2D RD map, the commonly used CFAR-CA detector determines the threshold by averaging the values of the 2D neighborhoods. However, our proposed pulse design achieves superior correlation performance along the range dimension, as illustrated by the RD map example in Fig. \ref{CFAR-CA}. This motivates us to design a hierarchical 1D CFAR-CA detector that first detects potential targets along the range dimension and then removes fake targets along the Doppler dimension. 

\subsubsection{CFAR-CA detection along the range dimension}
Let $\mathcal{M} = \{ \mathrm{P}(\hat{n},\hat{m}) \mid \mathrm{P}(\hat{n},\hat{m})$ is a local
maximum in the RD map$\}$ denote the set of all local maxima. For each $\mathrm{P}(\hat{n},\hat{m}) \in \mathcal{M}$, we apply the CFAR-CA detector along the delay/range domain, based on the following binary hypotheses:
\begin{equation}
    \begin{cases}
        \mathcal{H}_0: \mathrm{P}(\hat{n},\hat{m}) = z_{\hat{n}},\\[2pt]
        \mathcal{H}_1: \mathrm{P}(\hat{n},\hat{m}) = s + z_{\hat{n}},
    \end{cases}
\end{equation}
where $z_{\hat{n}}$ includes noise and RSI, while $s$ denotes the target echo component when present. The CFAR-CA decision rule is expressed as:
\begin{equation}
    \mathrm{P}(\hat{n},\hat{m}) \mathop{\gtrless}_{\mathcal{H}_0}^{\mathcal{H}_1} \eta_{\hat{n}},
\end{equation}
with the adaptive threshold:
\begin{equation}
    \eta_{\hat{n}} = \alpha_{\text{CFAR}} \hat{\sigma}_n^2, \quad
    \hat{\sigma}_n^2 = \frac{1}{|\Omega_n|}\sum_{k\in\Omega_n} \mathrm{P}(k,\hat{m}),
\end{equation}
where $\Omega_n$ denotes the set of training cells in the range domain, obtained by excluding the guard cells surrounding the cell-under-test (CUT), $|\Omega_n|$ is its cardinality, and $\alpha_{\text{CFAR}}$ is selected to achieve the desired false-alarm probability.

\subsubsection{CFAR-CA detection along the Doppler dimension}
We denote the potential target set after the first detection as $\mathcal{G} = \{ \mathrm{P}(\hat{n},\hat{m}) \in \mathcal{M} | \mathrm{P}(\hat{n},\hat{m}) > \eta_{\hat{n}} \}$.
For each candidate $\mathrm{P}(\hat{n},\hat{m}) \in \mathcal{G}$, we apply the CFAR-CA detector along the Doppler/velocity domain, performing the CFAR-CA decision:
\begin{equation}
    \mathrm{P}(\hat{n},\hat{m}) \mathop{\gtrless}_{\mathcal{H}_0}^{\mathcal{H}_1} \eta_{\hat{m}},
\end{equation}
with the adaptive threshold:
\begin{equation}
    \eta_{\hat{m}} = \alpha_{\text{CFAR}} \hat{\sigma}_m^2, \quad
    \hat{\sigma}_m^2 = \frac{1}{|\Omega_m|}\sum_{k\in\Omega_m} \mathrm{P}(\hat{n},k),
\end{equation}
where $\Omega_m$ denotes the set of training cells in the Doppler domain. With this step, some fake targets resulting from the Doppler sidelobes will be eliminated. Therefore, the final target detection result is $\mathcal{T} = \{ (\hat{n},\hat{m}) | \mathrm{P}(\hat{n},\hat{m}) \in \mathcal{G} \text{ and } \mathrm{P}(\hat{n},\hat{m}) > \eta_{\hat{m}} \}$. The detection algorithm is summarized in Algorithm \ref{alg:cfar_ca_detection}.

\begin{algorithm}[t]
\caption{ Hierarchical 1D CFAR-CA Detection}
\label{alg:cfar_ca_detection}
\begin{algorithmic}[1]
\REQUIRE 
    Range-Doppler map $\mathrm{P}$; CFAR factor $\alpha_{\text{CFAR}}$.
\ENSURE 
    Final target index set $\mathcal{T}$.
    
\STATE Extract the local maxima set from the RD map: $\mathcal{M} = \{ \mathrm{P}(\hat{n},\hat{m}) \}$.

\STATE \textit{\% Stage 1: CFAR-CA along the range dimension}
\STATE For each $\mathrm{P}(\hat{n},\hat{m}) \in \mathcal{M}$, compute the range adaptive threshold $\eta_{\hat{n}}$ using training cells $\Omega_n$.
\STATE Obtain the candidate target set: \\
       $\mathcal{G} = \{ \mathrm{P}(\hat{n},\hat{m}) \in \mathcal{M} \mid \mathrm{P}(\hat{n},\hat{m}) > \eta_{\hat{n}} \}$.
\STATE \textit{\% Stage 2: CFAR-CA along the Doppler dimension}
\STATE For each $\mathrm{P}(\hat{n},\hat{m}) \in \mathcal{G}$, compute the Doppler adaptive threshold $\eta_{\hat{m}}$ using training cells $\Omega_m$.
\RETURN The final detected target set: \\
        $\mathcal{T} = \{ (\hat{n},\hat{m}) \mid \mathrm{P}(\hat{n},\hat{m}) \in \mathcal{G} \text{ and } \mathrm{P}(\hat{n},\hat{m}) > \eta_{\hat{m}} \}$.
\end{algorithmic}
\end{algorithm}

\section{Optimization of Mismatched Filter}
\label{section 5}


According to Eq. (\ref{eq:receivedX}), the received signal structure depends on the delay of the target. 
As a result, the target detection at different delays is impacted by different factors, as illustrated in Fig.~\ref{performance metrics}. More specifically, for target delays satisfying $n_{\tau} \leq N_{\mathrm{r}}$ and $n_{\tau} \geq H+N_{\mathrm{r}}$, the detection performance is exclusively determined by the effective signal-to-interference-plus-noise ratio (SINR) or signal-to-noise ratio (SNR). 
This is attributed to the perfect cancellation of correlation sidelobes obtained via joint processing of the complementary and inverse-phase sequences.
In contrast, when $N_{\mathrm{r}}<n_{\tau} < H+N_{\mathrm{r}}$, the detection performance under the CFAR-CA detector is impacted by the sidelobes because of the partial eclipsing phenomenon. 
Consequently, the development of a new performance metric is essential for accurately characterizing the CFAR-CA detection performance for the target in the minimum range. In this section, we first establish the sensing performance metric for targets at different ranges and then optimize the weighting factor according to the proposed criterion.


\subsection{Sensing Performance Metrics} \label{III-C}



\begin{figure}[t]
    \centering
    \includegraphics[width=0.4\textwidth]{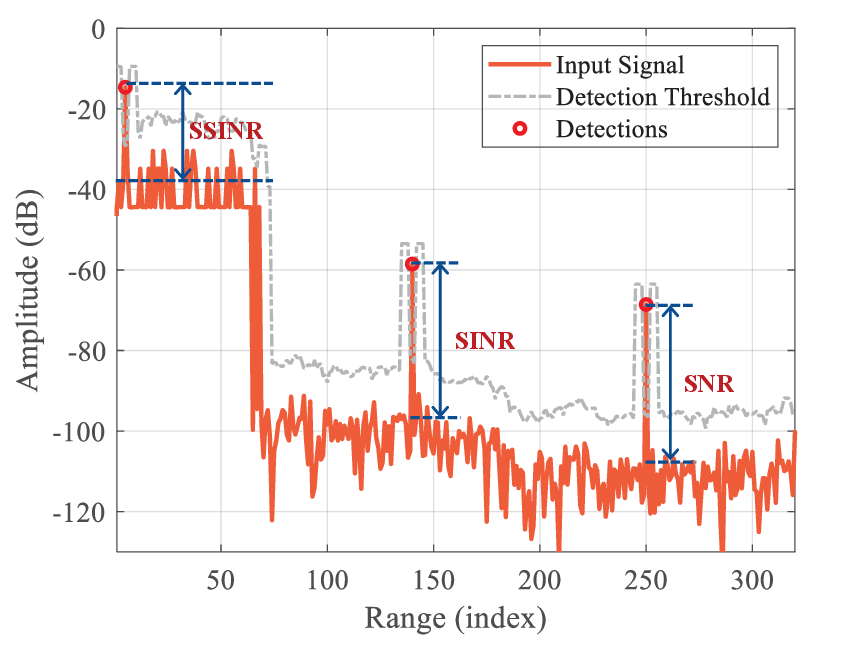 }
    \caption{Illustration of CFAR-CA detection and performance metrics.}
    \label{performance metrics}
    \vspace{-0em}
\end{figure}

For target delays satisfying $n_{\tau} \leq N_{\mathrm{r}}$ and $n_{\tau} \geq H+N_{\mathrm{r}}$, the correlation output at the  target delay $n_\tau$ can be expressed as
\begin{equation}
    \begin{aligned}
        &r_k(n_{\tau}) = \mathbf{f}_k^\mathsf{H} \mathbf{J}_{n_{\tau}} \mathbf{y}_k \\
                      &= \alpha e^{j 2 \pi f_d k T} \mathbf{f}_k^\mathsf{H} \mathbf{J}_{n_{\tau}} \mathbf{x}_{k,n_{\tau}} + \mathbf{f}_k^\mathsf{H} \mathbf{J}_{n_{\tau}} \beta \mathbf{P}_{\text{SI}} \mathbf{s}_k + \mathbf{f}_k^\mathsf{H} \mathbf{J}_{n_{\tau}}\mathbf{P}_{\text{z}} \mathbf{z}_k.
    \end{aligned}
\end{equation}
Therefore, the effective SINR/SNR of radar echoes over a CPI with $K$ pulses is given by
\begin{equation}
    \begin{aligned}
        F(n_\tau)
        &=\dfrac{K\|\alpha e^{j2\pi f_dkT}\mathbf{f}_k^{\mathsf{H}}\mathbf{J}_{n_{\tau}}\mathbf{x}_{k,n_{\tau}}\|^2}{\mathbb{E}[\|\mathbf{f}_k^\mathsf{H} \mathbf{J}_{n_{\tau}} \beta \mathbf{P}_{\text{SI}} \mathbf{s}_k + \mathbf{f}_k^\mathsf{H} \mathbf{J}_{n_{\tau}} \mathbf{P}_{\text{z}}\mathbf{z}_k\|^2]}\\
        &=\dfrac{K\|\alpha e^{j2\pi f_dkT}\mathbf{f}_k^{\mathsf{H}}\mathbf{J}_{n_{\tau}}\mathbf{x}_{k,n_{\tau}}\|^2}{\mathbb{E}\big[\|\mathbf{f}_k^{\mathsf{H}}\mathbf{J}_{n_{\tau}}\beta\mathbf{P}_{\text{SI}}\mathbf{s}_k\|^2\big]+\mathbb{E}\big[\|\mathbf{f}_k^{\mathsf{H}}\mathbf{J}_{n_{\tau}}\mathbf{P}_{\text{z}}\mathbf{z}_k\|^2\big]},
    \end{aligned}
\end{equation}
where $\mathbf{s}_k$ and $\mathbf{z}_k$ are statistically independent Gaussian random vectors. 

For target delays satisfying $N_{\mathrm{r}}<n_{\tau} < H+N_{\mathrm{r}}$, the detection performance is primarily governed by the peak-to-average sidelobe ratio (PASLR) $\gamma(n_\tau)$, which is defined to quantify the sidelobe levels produced by the partially received high-power component under the CFAR-CA detector. 
When applying CFAR-CA detection near boundaries where one-sided training cells are insufficient around the CUT, we need to dynamically redistribute training cells around the CUT while maintaining constant total training length. Consequently, PASLR $\gamma(n_\tau)$ can be mathematically characterized into three distinct operational cases.
When $N_{\mathrm{r}}+1\leq n_\tau\leq N_{\mathrm{r}}+\frac{g_\mathrm{c}}{2}+1$, this metric can be expressed as
\begin{equation}
    \gamma(n_\tau)=\dfrac{t_\text{c}\left(\displaystyle\sum_{k=0}^{K-1}r_{\mathrm{h},k}[n_\tau]\right)^2}{\displaystyle\sum_{n=n_\tau+\frac{g_\mathrm{c}}{2}+1}^{n_\tau+\frac{g_\mathrm{c}}{2}+t_\mathrm{c}}\left(\sum_{k=0}^{K-1}r_{\mathrm{h},k}[n]\right)^2},
\end{equation}
where \( g_\mathrm{c} \) and \( t_\mathrm{c} \) denote the number of guard cells and training cells in the 1D CFAR-CA detector, respectively; $r_{\mathrm{h},k}[n]$ represents the correlation value between the high-power sequence $\mathbf{h}_k^{\prime}=[\mathbf{h}_k^\mathsf{T},\mathbf{0}_{1 \times n_\tau}]^\mathsf{T}$ and its partially received sequence $\mathbf{h}_{k,n_\tau}=[\mathbf{0}_{1\times (H+N_{\mathrm{r}})
},h_{H-n_\tau+N_{\mathrm{r}},k},\dots,h_{H-1,k}]^\mathsf{T}$, which can be formally expressed as
\begin{equation}
    r_{\mathrm{h},k}[n]=\mathbf{h}_k^{\prime\mathsf{H}}\mathbf{J}_n\mathbf{h}_{k,n_\tau}.
\end{equation}
When $N_{\mathrm{r}}+\frac{g_\mathrm{c}}{2}+1< n_\tau< N_{\mathrm{r}}+\frac{g_\mathrm{c}}{2}+\frac{t_\mathrm{c}}{2}+1$, the PASLR $\gamma(n_\tau)$ can be expressed as
\begin{equation}
    \begin{aligned}
        &\gamma(n_{\tau})=\\
        &\frac{t_\mathrm{c}\left(\displaystyle\sum_{k=0}^{K-1}r_{\mathrm{h},k}[n_{\tau}]\right)^{2}}{\displaystyle\sum_{n=N_{\mathrm{r}}+1}^{n_{\tau}-\frac{g_\mathrm{c}}{2}-1}\left(\sum_{k=0}^{K-1}r_{\mathrm{h},k}[n]\right)^{2}+\sum_{n=n_{\tau}+\frac{g_\mathrm{c}}{2}+1}^{N_{\mathrm{r}}+g_\mathrm{c}+t_\mathrm{c}+1}\left(\sum_{k=0}^{K-1}r_{\mathrm{h},k}[n]\right)^{2}}.
    \end{aligned}
\end{equation}
When $N_{\mathrm{r}}+\frac{g_\mathrm{c}}{2}+\frac{t_\mathrm{c}}{2}+1\leq n_\tau<N_{\mathrm{r}}+H$, the PASLR $\gamma(n_\tau)$ can be expressed as
\begin{equation}
    \begin{aligned}
        &\gamma(n_{\tau})=\\
        &\frac{t_\mathrm{c}\left(\displaystyle\sum_{k=0}^{K-1}r_{\mathrm{h},k}[n_{\tau}]\right)^{2}}{\displaystyle\sum_{n=n_{\tau}-\frac{g_\mathrm{c}}{2}-\frac{t_\mathrm{c}}{2}}^{n_{\tau}-\frac{g_\mathrm{c}}{2}-1}\left(\sum_{k=0}^{K-1}r_{\mathrm{h},k}[n]\right)^{2}+\sum_{n=n_{\tau}+\frac{g_\mathrm{c}}{2}+1}^{n_{\tau}+\frac{g_\mathrm{c}}{2}+\frac{t_\mathrm{c}}{2}}\left(\sum_{k=0}^{K-1}r_{\mathrm{h},k}[n]\right)^{2}}
    \end{aligned}
\end{equation}

\begin{figure*}[t]
{\small
    \centering   
    \vspace{2pt}
        \begin{equation}\label{FuncF}
        \begin{aligned}
            F(n_\tau)=
            \begin{cases}
            \text{SINR}=\dfrac{K|\alpha|^2(P_\mathrm{l}L)^2}{|\beta|^2(L-n_\tau)P_\mathrm{l}^2+N_0BP_\mathrm{l}L},\hfill \forall n_\tau \in (0,N_{\mathrm{r}} ],
            \\
            \text{SSINR}=\dfrac{K|\alpha|^2\big(P_\mathrm{h}(n_\tau-N_{\mathrm{r}})+wP_\mathrm{l}L\big)^2}{\dfrac{K|\alpha|^2(P_\mathrm{h}(n_\tau-N_{\mathrm{r}}))^2}{\gamma(n_\tau)}+|\beta|^2\big((n_\tau-N_{\mathrm{r}}) P_\mathrm{h}P_\mathrm{l}+(L-n_\tau)w^2P_\mathrm{l}^2\big)+N_0B\big(P_\mathrm{h}(n_\tau-N_{\mathrm{r}})+w^2P_\mathrm{l}L\big)},\hfill \forall n_\tau \in (N_{\mathrm{r}},L ],
            \\
            \text{SSINR}=\dfrac{K|\alpha|^2\big(P_\mathrm{h}(n_\tau-N_{\mathrm{r}})+wP_\mathrm{l}L\big)^2}{\dfrac{K|\alpha|^2(P_\mathrm{h}(n_\tau-N_{\mathrm{r}}))^2}{\gamma(n_\tau)}+|\beta|^2\big((n_\tau-N_{\mathrm{r}}) P_\mathrm{h}P_\mathrm{l}\big)+N_0B\big(P_\mathrm{h}(n_\tau-N_{\mathrm{r}})+w^2P_\mathrm{l}L\big)},\hfill \forall n_\tau \in (L,L+N_{\mathrm{r}} ],
            \\
            \text{SSINR}=\dfrac{K|\alpha|^2\big(P_\mathrm{h}(n_\tau-N_{\mathrm{r}})+wP_\mathrm{l}L\big)^2}{\dfrac{K|\alpha|^2(P_\mathrm{h}(n_\tau-N_{\mathrm{r}}))^2}{\gamma(n_\tau)}+|\beta|^2LP_\mathrm{h}P_\mathrm{l}+N_0B\big(P_\mathrm{h}(n_\tau-N_{\mathrm{r}})+w^2P_\mathrm{l}L\big)},\hfill \forall n_\tau \in (L+N_{\mathrm{r}},H+N_{\mathrm{r}}),
            \\
            \text{SINR}=\dfrac{K|\alpha|^2\big(P_\mathrm{h}H+wP_\mathrm{l}L\big)^2}{|\beta|^2 (H +N_{\mathrm{r}}+ L - n_\tau) P_\mathrm{h} P_\mathrm{l} + N_0 B (P_\mathrm{h} H + w^2 P_\mathrm{l} L)},\hfill \forall n_\tau \in [H+N_{\mathrm{r}},H+N_{\mathrm{r}}+L ],
            \\
            \text{SNR}=\dfrac{K|\alpha|^2\big(P_\mathrm{h}H+wP_\mathrm{l}L\big)^2}{N_0 B (P_\mathrm{h} H + w^2 P_\mathrm{l} L)},\hfill \forall n_\tau \in (H+N_{\mathrm{r}}+L,S ],
            \\
            \text{SNR}=\dfrac{K|\alpha|^2\big(P_\mathrm{h}H+wP_\mathrm{l}(L+S-n_\tau)\big)^2}{N_0 B \big(P_\mathrm{h} H + w^2 P_\mathrm{l} (L+S-n_\tau)\big)},\hfill \forall n_\tau \in (S,L+S],
            \\
            \text{SNR}=\dfrac{K|\alpha|^2\big(P_\mathrm{h}H\big)}{N_0 B},\hfill \forall n_\tau \in (L+S,N_{\mathrm{r}}+L+S],
            \end{cases}
        \end{aligned}
        \end{equation}
    \vspace{4pt}
}
    \hrule
\end{figure*}

Accounting for sidelobe effects, we define the signal-to-sidelobe-plus-interference-plus-noise ratio (SSINR) for $N_{\mathrm{r}}<n_{\tau} < H+N_{\mathrm{r}}$. Considering all the cases, the complete expression for the CFAR-CA-based sensing performance metric $F(n_\tau)$ is derived as shown in Eq. (\ref{FuncF}) at the top of the next page.

\subsection{Determination of the Optimal $w$}
The target detection performance is determined by the sensing metric $F(n_\tau)$, which is a function of the tunable weight $w$ for $N_{\mathrm{r}}<n_{\tau} \leq L+S$. Consequently, we can optimize the weighting parameter $w$ of the mismatched filter for $N_{\mathrm{r}}<n_{\tau} \leq L+S $ to maximize this performance metric:
\begin{equation}
    \begin{aligned}
        \underset{w}{\text{maximize}} \quad & F(n_\tau;w)\\
    \end{aligned}
    \label{op}
\end{equation}
It is obvious that when $w=1$, the mismatched filter degenerates to conventional matched filtering, and when $w=0$ ($w \to \infty$), only the high-power (low-power) sequence is utilized for sensing. The intermediate values of $w$ generate weighted combinations of both sequences' correlation outputs, establishing a continuous trade-off space. 
Since it is a single-variable optimization problem, the optimal value can be achieved by:
\begin{equation}
    \left.\frac{\partial F(n_\tau; w)}{\partial w}\right|_{w=w^*} = 0.
    \label{NOC}
\end{equation}
Based on Eq. (\ref{NOC}), it is easy to identify that for $n_{\tau} > H+N_{\mathrm{r}}+L$, the optimal value is $w=1$, which corresponds to the matched filter. 

For $N_{\mathrm{r}} <n_{\tau} \leq H+N_{\mathrm{r}}$, the optimal $w$ for the mismatched filter also depends on other parameters, as analyzed below. For $n_\tau \in (N_{\mathrm{r}}, L]$, the optimal $w$ is given by:
\begin{equation}
    w^*(n_\tau,\sigma) = \dfrac{\dfrac{KL|\alpha|^2P_{\text{h}}(n_\tau-N_{\mathrm{r}})}{\gamma(n_\tau)} + L|\beta|^2P_{\text{l}} + N_0BL}{|\beta|^2(L-n_\tau)P_{\text{l}} + N_0BL}.
    \label{op1}
\end{equation}
We can see that the optimal weight $w^*(n_\tau,\sigma)$ is a function of both the delay $n_\tau$ and the RCS $\sigma$. However, RCS is not a systematic parameter and the ISAC BS has no a priori knowledge of the RCS for a target. Thus, to determine the optimal value of $w^*$, we need to find a proper value of $\sigma$. To this end, we substitute Eq. (\ref{op1}) into the SSINR formulation, and thus we have $\text{SSINR} = f(n_\tau,\sigma)$, which is a function of $\sigma$. With $f(n_\tau,\sigma)$, we want to find the minimum detectable RCS $\sigma^*$ that can support reliable detection. 

\newtheorem{proposition}{Proposition}
\begin{proposition} \label{prop:monotonicity}
For any given delay bin $n_\tau \in (N_\mathrm{r}, L]$, the function $f(n_\tau,\sigma)$ is strictly monotonically increasing with respect to $\sigma$ for all $\sigma > 0$. Consequently, for any prescribed detection threshold $\rho > 0$, there exists a unique solution $\sigma^* > 0$ satisfying $f(n_\tau,\sigma^*) = \rho$.
\end{proposition}
\begin{IEEEproof}
The detailed proof is provided in Appendix \ref{app:proof_monotonicity}.
\end{IEEEproof}

Denoting the minimum detectable SNR by $\rho$, Proposition \ref{prop:monotonicity} establishes $\sigma^*$ as the corresponding minimum detectable RCS.
In practice, since the exact RCS of a target is unknown, the system is configured using the weight optimized for this lower bound, i.e., $w = w(n_\tau, \sigma^*)$.
Under this weight design, since the channel gain $|\alpha|^2$ is proportional to $\sigma$, the achieved SSINR, denoted by $f(n_\tau, \sigma \mid w^*(n_\tau, \sigma^*))$, maintains strict monotonic growth with respect to the actual target RCS $\sigma$.
Consequently, any target with an RCS $\sigma \ge \sigma^*$ will automatically satisfy the detection condition $f(n_\tau, \sigma \mid w^*(n_\tau, \sigma^*)) \ge \rho$. This property demonstrates that $w(n_\tau, \sigma^*)$ serves as the robust and optimal weight solution when the exact target RCS is unavailable at the ISAC BS.
Therefore, for any delay $n_\tau \in (N_{\mathrm{r}}, L]$, the optimal solution is finally formulated as
\begin{equation}
    w^*(n_\tau) = \dfrac{\dfrac{KL{G_\mathrm{t}G_\mathrm{r}\lambda^2\sigma^*}P_{\text{h}}(n_\tau-N_{\mathrm{r}})}{\gamma(n_\tau){(4\pi)^3(n_\tau T_{\mathrm{p}})^4}} + L|\beta|^2P_{\text{l}} + N_0BL}{|\beta|^2(L-n_\tau)P_{\text{l}} + N_0BL}.
\end{equation}
For delay $n_\tau \in (L, H+N_{\mathrm{r}})$, the optimal $w$ can be derived through analogous mathematical procedures.

For $n_\tau \in [H+N_{\mathrm{r}}, H+N_{\mathrm{r}}+L]$, the optimal $w$ can be directly derived from Eq. (\ref{NOC}) as 
\begin{equation}
    w^*(n_\tau)=\frac{|\beta|^2(H+N_{\mathrm{r}}+L-n_\tau)P_\mathrm{l}}{N_0BH}+1,
\end{equation}
where the optimal $w^*(n_\tau)$ relies entirely on the delay bin $n_\tau$, effectively decoupling the weight design from the unknown target RCS.


Based on the preceding analysis, the $w^*(n_\tau)$ is derived as
{\small
\begin{equation}
    \begin{aligned}
         &w^*(n_\tau)=\\
         &\begin{cases}                            \dfrac{\frac{KL{G_\mathrm{t}G_\mathrm{r}\lambda^2\sigma^*}P_{\text{h}}(n_\tau-N_{\mathrm{r}})}{\gamma(n_\tau){(4\pi)^3(n_\tau T_{\mathrm{p}})^4}} + L|\beta|^2P_{\text{l}} + N_0BL}{|\beta|^2(L-n_\tau)P_{\text{l}} + N_0BL},\hfill n_\tau \in (N_{\mathrm{r}}, L],\\
         \dfrac{\frac{K{G_\mathrm{t}G_\mathrm{r}\lambda^2\sigma^*}P_{\text{h}}(n_\tau-N_{\mathrm{r}})}{\gamma(n_\tau){(4\pi)^3(n_\tau T_{\mathrm{p}})^4}} + |\beta|^2P_{\text{l}} + N_0B}{N_0B},\hfill n_\tau \in (L, L+N_{\mathrm{r}}],\\
         \frac{\frac{KG_\mathrm{t}G_\mathrm{r}\lambda^2\sigma^*(P_{\text{h}}(n_\tau-N_{\mathrm{r}}))^2}{\gamma(n_\tau)(4\pi)^3(n_\tau T_{\mathrm{p}})^4} + |\beta|^2LP_\text{h}P_\text{l}}{ N_0BP_{\text{h}}(n_\tau-N_{\mathrm{r}})}+1, \hfill n_\tau \in (L+N_{\mathrm{r}}, H+N_{\mathrm{r}}),\\
         \frac{|\beta|^2(H+N_{\mathrm{r}}+L-n_\tau)P_\mathrm{l}}{N_0BH}+1, \hfill n_\tau \in [H+N_{\mathrm{r}}, H+N_{\mathrm{r}}+L],\\
         1,\hfill  n_\tau \in (H+N_{\mathrm{r}}+L,L+S].\\
         \end{cases}
    \end{aligned}
    \label{optimalw}
\end{equation}
}

This result provides a closed-form solution for the optimal weight $w^*$, balancing the contributions of high-power and low-power sequences while maximizing the SSINR under the given radar system constraints.


\section{Performance Evaluation} \label{section 6}

The simulation setup of the OFDM system follows 3GPP standards \cite{38211}, and the parameters are listed in Table \ref{tab:para}.
Since the recovery time is generally insignificant relative to the pulse duration, i.e., $T_{\mathrm{r}}\ll T_{\mathrm{h}}$, we ignore this transition time in our simulations. It should be noted that the proposed waveform design and associated signal processing techniques remain fully applicable to systems with a recovery time.
The duration of the sensing symbol is the same as the duration of the data symbol to keep the compatibility of the frame structure. The sensing symbol is inserted every 14 symbols so that the sensing overhead is  $\frac{1}{14} = 7\%$. 

\begin{table}[t]
    \centering
    \caption{Simulation Parameters for Sensing}
    \label{tab:para}
    \begin{tabular}{@{}cc}
    \toprule
    \textbf{Parameter} & \textbf{Value}  \\
    \midrule
    Carrier frequency  $f_\mathrm{c}$ & 28 GHz\\
    Bandwidth  $B$ & 100 MHz  \\
    Pulse repetition interval $T$ & 0.125 ms\\
    Number of PRIs per CPI $K$ & 32\\
    Transmit power $P_{\mathrm{h}}$ & 53 dBm  \\
    Transmit power $P_{\mathrm{l}}$ & 35 dBm  \\
    Symbol duration $T_{\mathrm{t}}$ & 8.92 $\mu$s  \\
    high-power waveform duration $T_\mathrm{h}$ & 1.28 $\mu$s \\
    low-power waveform duration $T_\mathrm{l}$ & 0.64 $\mu$s \\
    Length of high-power sequence $H$ & 128\\
    Length of low-power sequence $L$ & 64\\
    Noise PSD $N_0$ & -174 dBm/Hz\\
    Receiver noise figure & 5 dB\\
    TX \& RX antenna gain $G_\mathrm{t}$ \& $G_\mathrm{r}$ & 20 dBi\\
    Target RCS $\sigma$ & -10 dBsm \\
    Probability of false alarm $P_\mathrm{FA}$ & $10^{-5}$\\
    Chip duration $T_\text{p}$ & 0.01$\mu$s\\
    Minimum detectable SNR $\rho$ & 15 dB\\
    \bottomrule
    \end{tabular}
    \vspace{-1em}
\end{table}

\begin{figure*}[t]
    \centering
    \subfigure[Detection probability with SIC = 100 dB]
    {\includegraphics[width = 0.32\linewidth]{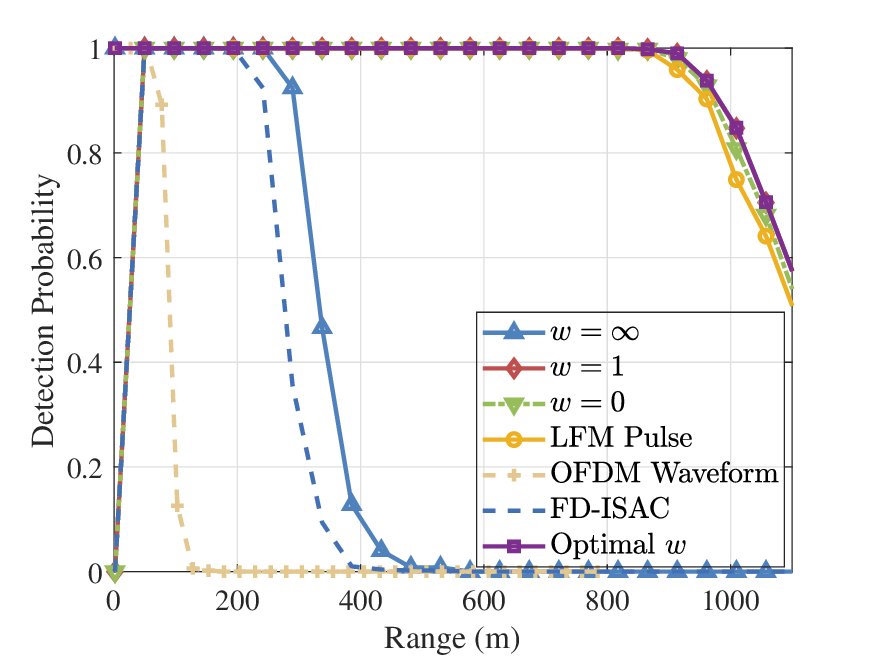}\label{PdSIC=100}}
    \subfigure[Detection probability with SIC = 110 dB]
    {\includegraphics[width=0.32\textwidth]{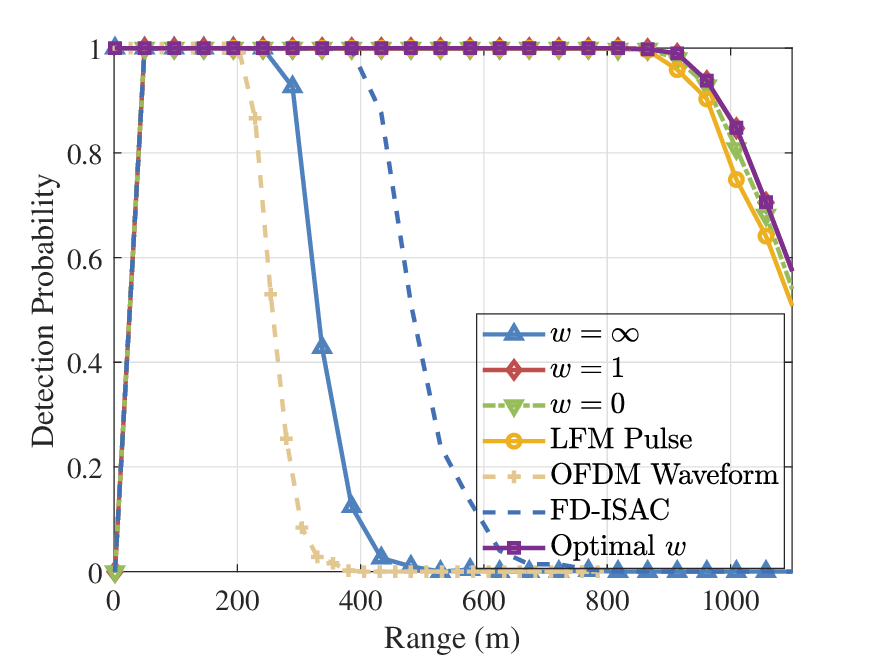}\label{PdSIC=110}}
    \subfigure[Detection probability with SIC = 120 dB]
    {\includegraphics[width=0.32\linewidth]{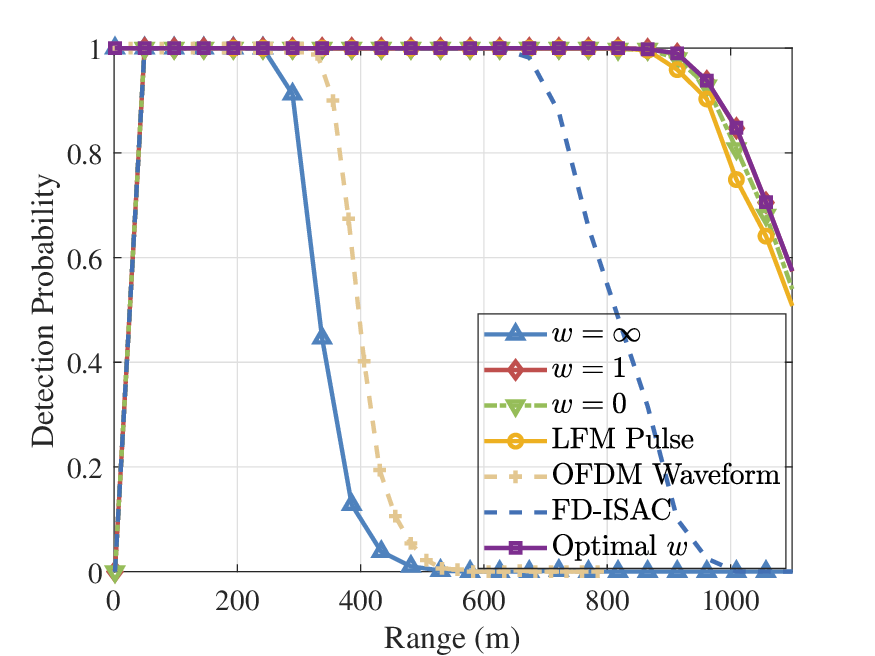}\label{PdSIC=120}}
    \caption{Detection probability under different SIC levels.
}\label{Pd}
    \vspace{-0em}
\end{figure*}

\begin{figure*}[t]
    \centering
    \subfigure[Minimum Detectable RCS with SIC = 100 dB]
    {\includegraphics[width = 0.32\linewidth]{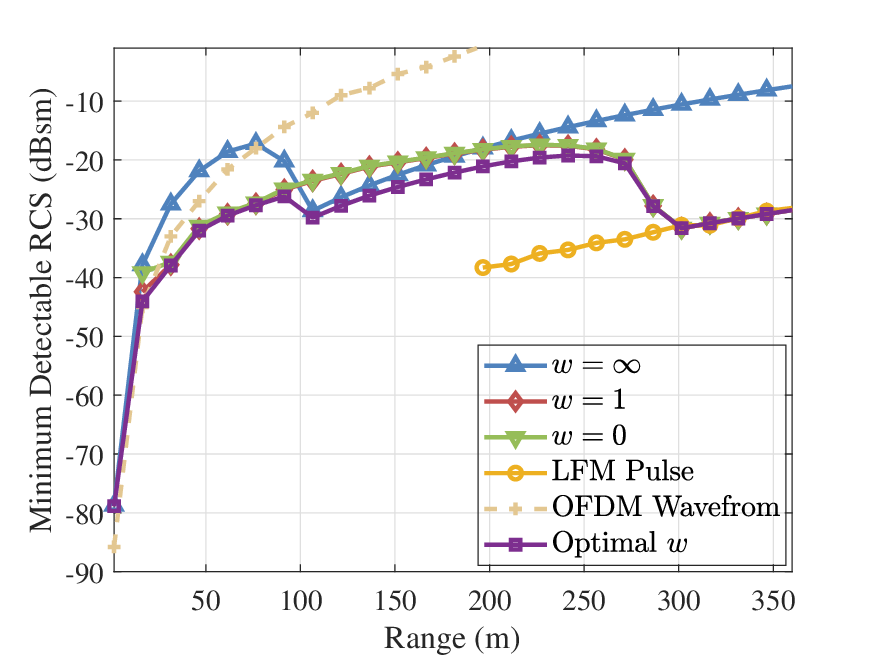}\label{minRCS_SIC=100}}
    \subfigure[Minimum Detectable RCS with SIC = 110 dB]
    {\includegraphics[width=0.32\textwidth]{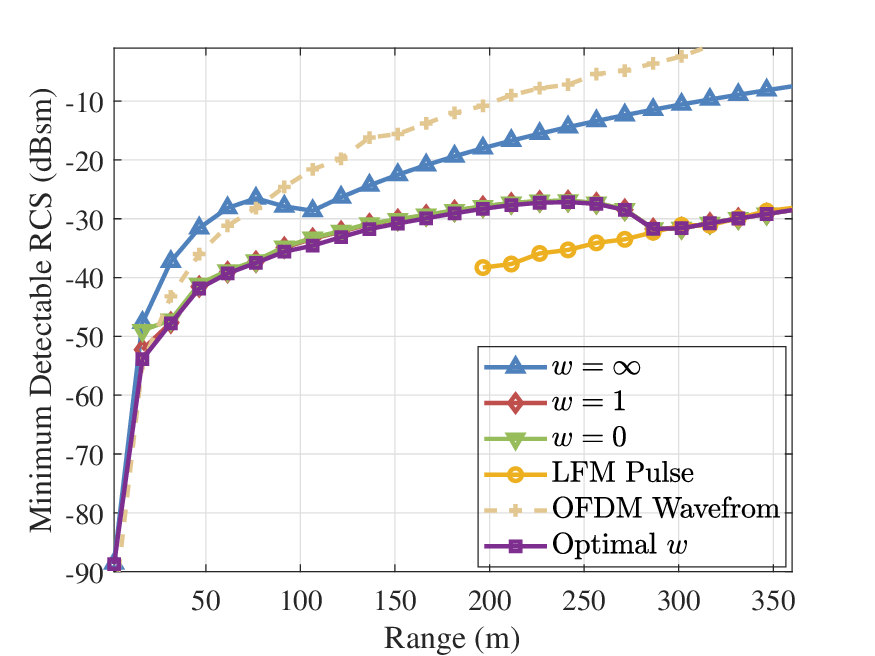}\label{minRCS_SIC=110}}
    \subfigure[Minimum Detectable RCS with SIC = 120 dB]
    {\includegraphics[width=0.32\linewidth]{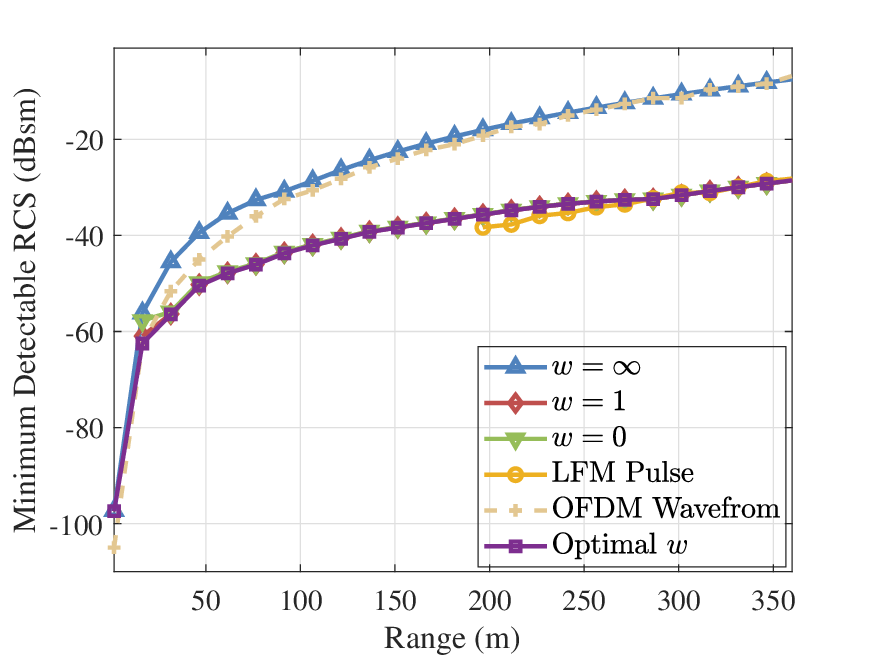}\label{minRCS_SIC=120}}
    \caption{Minimum detectable RCS for short-range targets under different SIC levels.
}\label{minRCS}
    \vspace{-0em}
\end{figure*}

The detection probability curves of the proposed waveform under different SIC ratios with different weighting factors are shown in Fig. \ref{Pd}. The comparison is conducted for three sensing waveforms: the LFM pulse, the continuous OFDM waveform, and the pulse waveform design in \cite{waveform2022} that employs continuous communication-signal transmission during conventional radar's silent periods, and the receiver is deactivated during high-power pulse transmission and enabled during communication-signal transmission.

As shown in Fig. \ref{PdSIC=100}, regardless of only using high-power sequence ($w=0$), or matched filtering ($w=1$), both methods exhibit unreliable detection regions in the short range, resulting in a detection blind range. Furthermore, when relying solely on low-power signals for sensing ($w=\infty$), the limited energy restricts the detection range. In contrast, the proposed waveform with the optimal $w$ design achieves comprehensive coverage from short-range to long-range detection. As illustrated in Fig. \ref{PdSIC=110} and Fig. \ref{PdSIC=120}, the proposed sensing waveform demonstrates similar advantages under high SIC conditions.
Compared to the LFM pulse, the proposed waveform achieves a slightly longer maximum detection range, benefiting from the additional energy provided by the low-power signal. Compared to the OFDM waveform, the detection range is significantly enlarged. Moreover, the performance of the OFDM waveform highly relies on the SIC capability. Nevertheless, the proposed waveform maintains reliable detection performance for both short-range and long-range targets, even under low SIC conditions. 
This robustness follows from the fact that the proposed waveform preserves the silent period of conventional pulsed radars. 

\begin{figure*}[t]
    \centering
    \subfigure[short-range case]
    {\includegraphics[width = 0.41\linewidth]{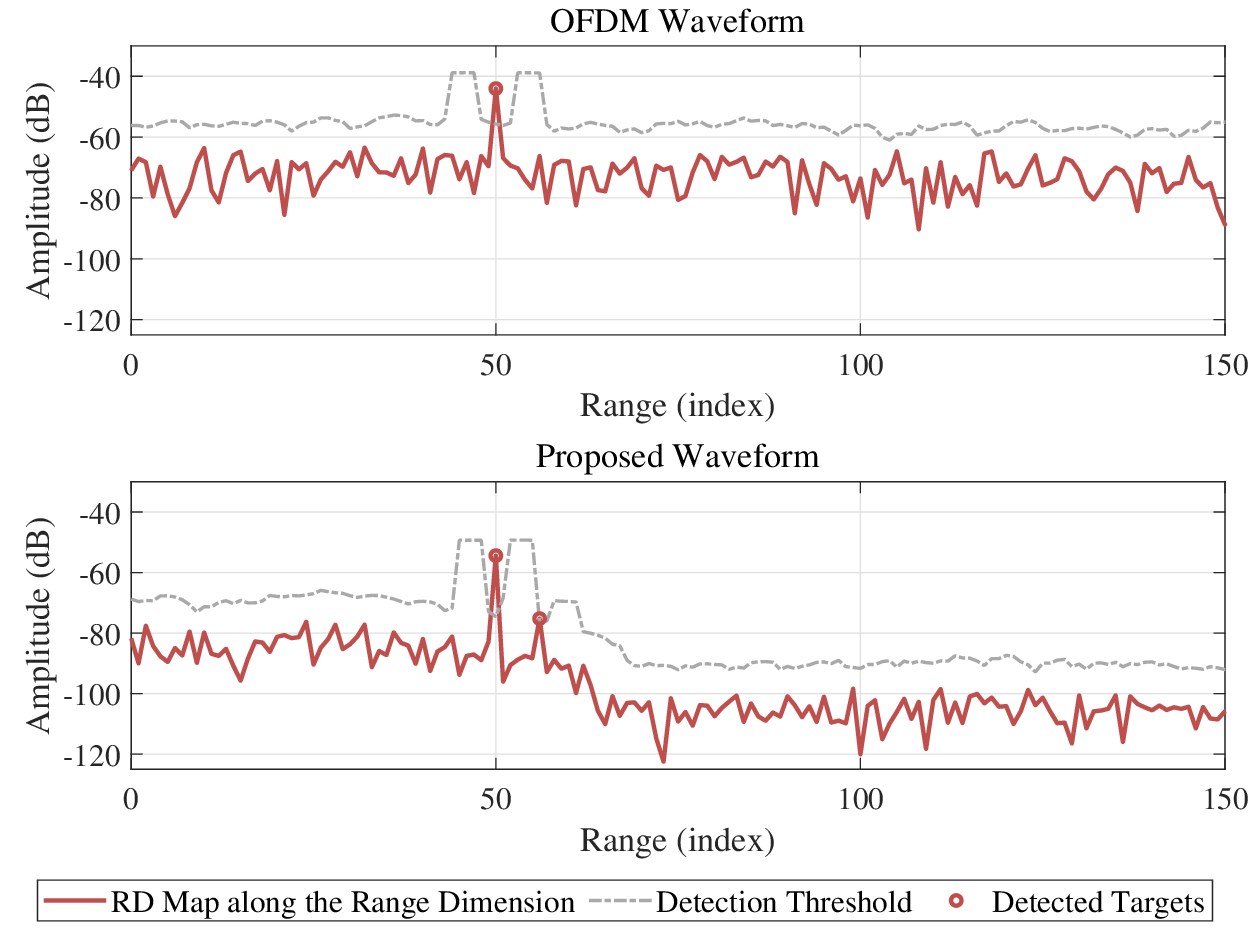}\label{fig: MOFDM}}
    \subfigure[long-range case]
    {\includegraphics[width=0.58\textwidth]{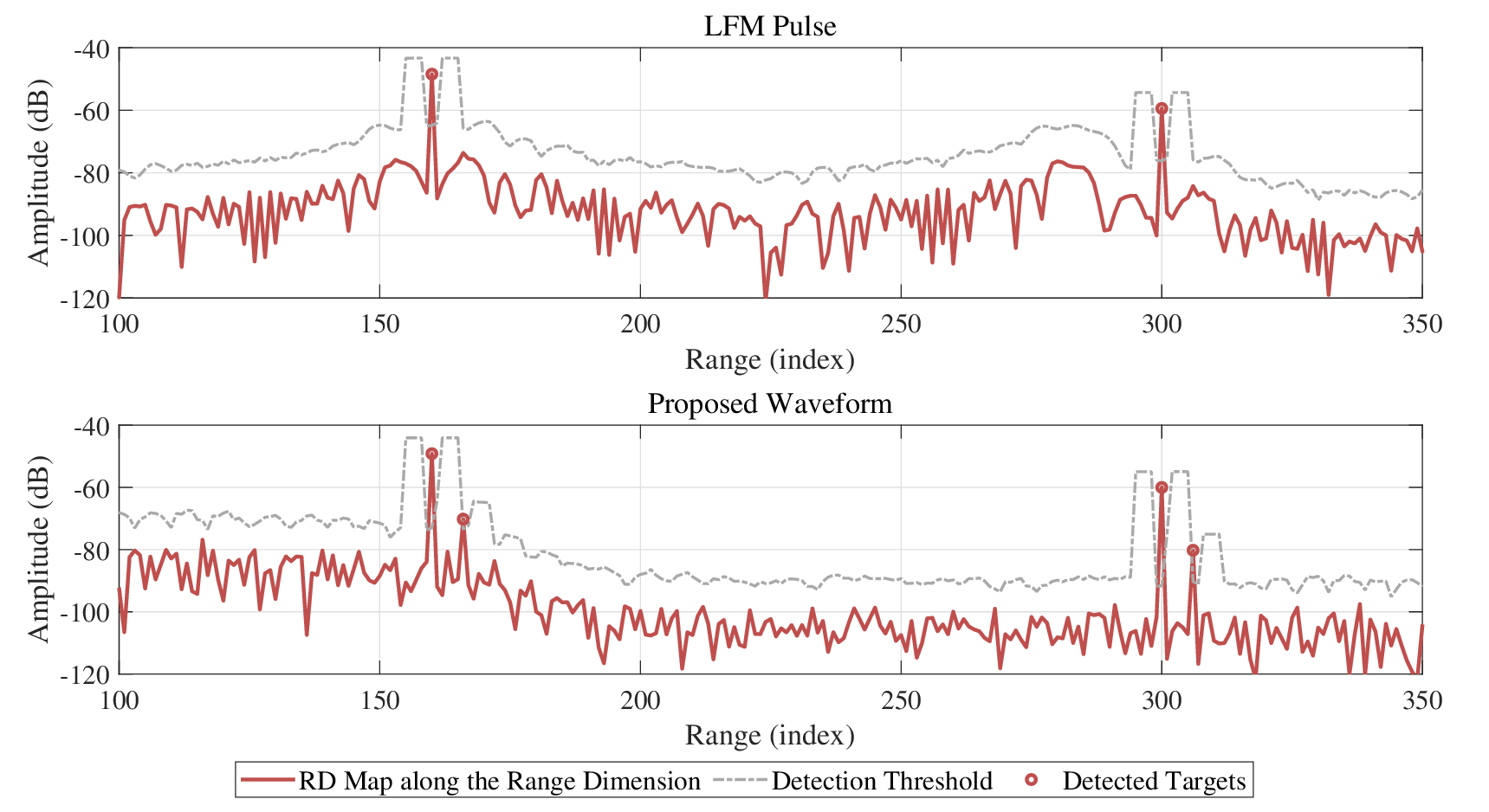}\label{fig: MLFM}}
    \vspace{-0em}
    \caption{Multi-target detection performance: (a) comparison between proposed waveform and OFDM waveform with $n_\tau=[50, 56]$ index and $\sigma=[0, -20]$ dBsm; (b) comparison between proposed waveform and LFM pulse with $n_\tau=[160, 166,300,306]$ index and $\sigma=[0, -20,0,-20]$ dBsm.
}\label{fig:multiTarget}
\end{figure*}

\begin{figure}[t]
    \centering
    \includegraphics[width=0.8\linewidth]{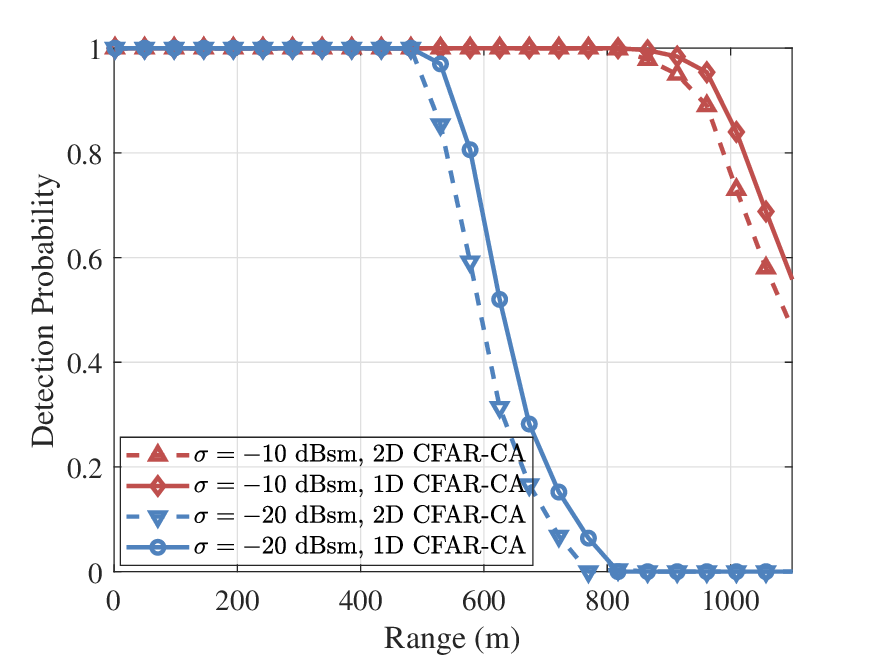}
    \caption{Performance comparison of different CFAR-CA detection methods.}
    \label{CFAR-CAcompare}
    \vspace{-1em}
\end{figure}
Since the low-power component in our proposed design introduces RSI for short-range targets, its impact is further evaluated through the minimum detectable RCS, as shown in Fig. \ref{minRCS}.
For our proposed design, the minimum detectable RCS with the optimal $w$ consistently lies below that of other filtering approaches employing different $w$ values, which demonstrates that a tunable $w$ for the mismatched filter can enhance the overall detection capability across all ranges. For $w=0$ and $w=1$, no minimum detectable RCS exists within the first 20 meters because the PSLR is insufficient for target detection. 
For the optimal $w$, when $R<20$ m, our design performs worse than the OFDM waveform. 
However, the minimum detectable RCS attains a sensitivity level below $-40$ dBsm under these SIC conditions, which is sufficient to meet the requirements of most ISAC sensing applications. 
For the optimal $w$, in the range $R>20$ m, our design outperforms the OFDM waveform. As the SIC level decreases, the minimum detectable RCS increases for both the OFDM waveform and our proposed design. However, the increase is significantly smaller for our waveform, demonstrating its superior resistance to self-interference. When the target range is within 192 m $<R<$ 300 m, the echo signal of the proposed waveform is affected by RSI, making its detection performance dependent on the SIC capability, whereas the LFM pulse remains unaffected by RSI. Consequently, our design exhibits inferior detection performance compared to the LFM pulse. For the target range $R$ exceeding 300 meters, where both the proposed waveform and conventional LFM pulses are unaffected by RSI, the proposed design achieves a slightly lower minimum detectable RCS compared to the conventional LFM pulse due to the additional energy provided by the low-power component. 
In short, our proposed design achieves complete elimination of the blind range inherent to conventional radar, while maintaining long-range detection performance independent of the SIC capability. 


Next, the multi-target detection performance is evaluated by using two closely spaced targets but having different RCSs of 0 dBsm and -20 dBsm. 
We employ the OFDM waveform as the performance benchmark for short-range detection, while utilizing the LFM pulse waveform for long-range detection. 
As shown in Fig. \ref{fig: MOFDM}, the multi-target detection is conducted for both the OFDM waveform and our proposed waveform for targets in the short range. 
For the OFDM waveform, the RSI masks the mainlobe of the weaker target, resulting in only a single detectable target.
In contrast, our proposed waveform effectively suppresses this interference, allowing simultaneous detection of both strong and weak targets. 
Another simulation is conducted for targets in long range with the LFM pulse and our proposed waveform, as shown in Fig. \ref{fig: MLFM}. The conventional LFM implementation exhibits performance degradation due to the sidelobe interference, where the strong target's sidelobes completely mask the weak target's mainlobe. Our waveform design overcomes this limitation through the combination of complementary sequences and the mismatched filtering technique. This innovative approach enables simultaneous detection of multiple targets with different signal strengths in one single processing iteration, achieving a substantial reduction in computational complexity compared to traditional successive target cancellation algorithms while maintaining detection accuracy. These results confirm that our solution provides more reliable multi-target detection capability.


The performance of our proposed hierarchical 1D CFAR-CA detector is also evaluated with respect to the maximum detectable range.   
As illustrated in Fig. \ref{CFAR-CAcompare}, for targets with identical RCS, the maximum detectable range achieved by the 2D CFAR-CA algorithm is shorter than that by our proposed hierarchical 1D CFAR-CA approach. This performance difference originates from the complementary sequence design that generates an RD map with perfect range-domain sidelobe cancellation. Consequently, the hierarchical 1D CFAR-CA detector operates without range sidelobe interference, while the 2D CFAR-CA detector remains susceptible to Doppler-domain sidelobe effects. As a result, a lower but reliable detection threshold is achieved by our proposed detector to support longer range detection.


\section{Conclusion} \label{section 7}


In this paper, a novel ISAC waveform was developed for 6G LAWN to address the limitations of conventional sensing approaches under imperfect full-duplex radios. A key innovation was the dual-power phase-coded pulse design, consisting of a high-power leading sequence and a low-power trailing sequence, which successfully bridged the gap between half-duplex long-range capabilities and full-duplex blind-range elimination. Furthermore, complementary and inverse-phase sequences were developed to ensure robust autocorrelation properties and sidelobe suppression for improving multi-target detection. Based on the proposed signal structure, a mismatched filter with the optimized parameter was designed to enhance the target detection performance, where a hierarchical 1D CFAR-CA detection mechanism was designed to exploit the benefits of our proposed waveform. Simulation results validated significant improvements in maximum detection range, minimal detectable RCS in the short range, and multi-target discrimination capability. The proposed waveform design established a viable solution framework for addressing the fundamental sensing challenge of imperfect full-duplex radios in emerging 6G ISAC systems.
\appendices

\section{Proof of Proposition \ref{prop:monotonicity}}
\label{app:proof_monotonicity}

To prove the strict monotonicity of $f(n_\tau,\sigma)$ with respect to $\sigma$, we define an intermediate variable $x = \frac{G_\mathrm{t}G_\mathrm{r}\lambda^2\sigma}{(4\pi)^3R^4}$. It is mathematically equivalent to prove $f'(x) > 0$ for $x > 0$.

For a given $n_\tau$, we introduce the following auxiliary positive constants: $A = P_\mathrm{h}(n_\tau - N_\mathrm{r})$, $D = P_\mathrm{l} L$, $F= |\beta|^2(L - n_\tau)P_\mathrm{l} + N_0 B L$. Furthermore, we define $m = \frac{K L A}{\gamma(n_\tau) F}$, $n = \frac{L|\beta|^2 P_\mathrm{l} + N_0 B L}{F}$, $c = \frac{K A^2}{\gamma(n_\tau)}$, $d = |\beta|^2(n_\tau - N_\mathrm{r})P_\mathrm{h} P_\mathrm{l} + N_0 B A$, and $e = P_\mathrm{l} F$.

Substituting the linear optimal weight $w^*(x) = m x + n$ into the SSINR formulation, $f(x)$ is constructed as a rational function:

{\small
\begin{equation}
    f(x) = \frac{K x \left(A + (mx+n)D\right)^2}{c x + d + e(mx+n)^2}.
\end{equation}    
}

Taking the first derivative yields $f'(x) = g(x)/D(x)^2$, where the numerator is a quartic polynomial
$g(x) = \sum_{i=0}^{4} C_i x^i$.
By expanding $f(x)$, one can verify that the coefficients $C_0, C_1, C_3$, and $C_4$ consist entirely of positive additive terms, making them strictly positive. The critical step is to determine the sign of the quadratic coefficient $C_2$, which can be formulated as:

{\small
\begin{equation} \label{eq:C2_medium}
    \frac{C_2}{K m^2} = \underbrace{3 D^2 d + 6 D^2 e n^2 + \frac{2 D^2 n c}{m} + 2 A D e n}_{>0} +\left( \frac{2 D A c}{m} - A^2 e \right).
\end{equation}
}

To evaluate the sign of the bracketed term, we substitute the parameters back into the expression. Noting that $\frac{c}{m} = \frac{A F}{L}$, $D = P_\mathrm{l} L$, and $e = P_\mathrm{l} F$, we obtain:
\begin{equation}
    \frac{2 D A c}{m} - A^2 e = 2 P_\mathrm{l} F A^2 - P_\mathrm{l} F A^2 = P_\mathrm{l} F A^2 > 0.
\end{equation}
Since this bracketed term evaluates to a strictly positive value, it follows from \eqref{eq:C2_medium} that $C_2 > 0$. Consequently, all coefficients of $g(x)$ are strictly positive. For any $\sigma > 0$ (i.e., $x > 0$), $g(x) > 0$ holds, which guarantees $f'(\sigma) > 0$. This completes the proof.

\bibliographystyle{IEEEtran}
\bibliography{reference}

\end{document}